\documentclass[11pt]{article}
\usepackage{amsmath,amsthm,amssymb,hyperref, color,fullpage}
\newcommand{\remove}[1]{}
\hypersetup{
	colorlinks=true,
	linkcolor=blue,%
	citecolor=blue,
	linktoc=page
}

\newtheorem{thm}{Theorem}[section]
\newtheorem{claim}[thm]{Claim}
\newtheorem{lem}[thm]{Lemma}
\newtheorem{define}[thm]{Definition}
\newtheorem{cor}[thm]{Corollary}

\newtheorem{example}[thm]{Example}

\newtheorem{construction}[thm]{Construction}

\def\F{{\mathbb{F}}}

\def\R{{\mathbb{R}}}

\def\cS{{\cal S}}

\def\cC{{\cal C}}
\def\cF{{\cal F}}
\def\cO{{\cal O}}
\def\cM{\mathcal{M}}

\def\cS{{\mathcal S}}
\def\cE{{\mathcal E}}

\def\E{{\mathbb E}}

\def\_{\,\,\,\,\,}
\def\prob{{\mathbf{Pr}}}

\def\reach{\textsf{Rea}}

\def\rank{\textsf{rank}}
\def\poly{\textsf{poly}}

\def\dim{\textsf{dim}}

\def\spn{\textsf{span}}
\def\idim{\textsf{idim}}
\def\lt{\ltimes}
\def\inj{\hookrightarrow}
\def\es{\emptyset}
\def\frank{{\textsf{f-rank}}}
\def\LPcover{{\textsf{LP}^{cover}}}
\def\LPentropy{{\textsf{LP}^{entropy}}}

\def\vc{\textsf{VC-dim}}

\def\DG{{\cal D}}

\def\spans{{\, \models \, }}

\newcommand{\eps}{\epsilon}

\usepackage{pgf,tikz,pgfplots}
\usepackage{mathrsfs}
\usetikzlibrary{arrows}

\begin{document}

\title{{\em Spanoids} - an abstraction of spanning structures, and a barrier for LCCs}

\date{}

\author{Zeev Dvir\thanks{Department of Computer Science and Department of Mathematics,
Princeton University.
Email: \texttt{zeev.dvir@gmail.com}. Research supported by NSF CAREER award DMS-1451191 and NSF grant CCF-1523816.} \and
Sivakanth Gopi\thanks{Microsoft Research.
Email: \texttt{sigopi@microsoft.com}. Research supported by NSF CAREER award DMS-1451191 and NSF grant CCF-1523816.} \and
Yuzhou Gu\thanks{MIT.
Email: \texttt{yuzhougu@mit.edu}. Research supported by Jacobs Family
Presidential Fellowship.}\and
Avi Wigderson\thanks{Institute for Advanced Study.
Email: \texttt{avi@math.ias.edu}. Research was partially supported by NSF grant CCF-1412958.}
}

\maketitle

\begin{abstract}

We introduce a simple logical inference structure we call a {\em spanoid} (generalizing  the notion of a matroid), which captures well-studied problems in several areas. These include combinatorial geometry (point-line incidences), algebra (arrangements of hypersurfaces and ideals), statistical physics  (bootstrap percolation), network theory (gossip / infection processes) and coding theory.  We initiate a thorough investigation of spanoids, from computational and structural viewpoints, focusing on parameters relevant to the applications areas above and, in particular, to questions regarding Locally Correctable Codes (LCCs). 

One central parameter we study is the {\em rank} of a spanoid, extending the rank of a matroid and related to the dimension of codes. This leads to one main application of our work, establishing the first known barrier to improving the nearly 20-year old bound of Katz-Trevisan (KT) on the dimension of LCCs. On the one hand, we prove that the KT bound (and its more recent refinements) holds for the much more general setting of spanoid rank. On the other hand we show that there exist (random) spanoids whose rank matches these bounds. Thus, to significantly improve the known bounds one must step out of the spanoid framework.

Another parameter we explore is the {\em functional rank} of a spanoid, which captures the possibility of turning a given spanoid into an actual code. The question of the relationship between rank and functional rank is one of the main questions we raise  as it may reveal new avenues for constructing new LCCs (perhaps even matching the KT bound). As a first step, we develop an entropy relaxation of functional rank to create a small constant gap and amplify it by tensoring to construct a spanoid whose functional rank is smaller than rank by a polynomial factor. This is evidence that the entropy method we develop can prove polynomially better bounds than KT-type methods on the dimension of LCCs.

To facilitate the above results we also develop some basic structural results on spanoids including an equivalent formulation of spanoids as set systems and properties of spanoid products. We  feel that given these initial findings and their motivations, the abstract study of spanoids merits further investigation. We leave plenty of concrete open problems and directions.
\end{abstract}
\thispagestyle{empty}
\newpage
\tableofcontents
\thispagestyle{empty}
\newpage
\setcounter{page}{1}
 \section{Introduction}

This (somewhat long) introduction will be organized as follows. We begin by discussing Locally Correctable Codes (LCCs) and the main challenges they present as this was the primary motivation for this work. We proceed to define spanoids as an abstraction of LCCs, and state some results about their rank which hopefully illuminate the difficulties with LCCs in a new light. We continue by describing other natural settings in which the spanoid structure arises in the hope of motivating the questions raised in the context of LCCs and demonstrating their potential to contribute to research in other areas. We then turn to the investigation of functional rank of spanoids, which aims to convert them to actual LCCs. We conclude with describing some of the structural results about spanoids obtained here.

\subsection{Locally Correctable Codes} \label{sec-intro-lccs}
 
 The introduction of {\em locality} to coding theory has created a large body of research with wide-ranging applications and connections, from probabilistically checkable proofs, private information retrieval, program testing, fault-tolerant storage systems, and many others in computer science and mathematics. We will not survey these, and the reader may consult the surveys \cite{Yekhanin12,Dvir-incidences}. Despite much progress, many  basic questions regarding local testing, decoding and correcting of codes remain open. Here we focus on the efficiency of  locally {\em correctable} codes, that we now define. Note that the related, locally {\em decodable} codes (LDCs), will not be discussed in this paper, as our framework is not relevant to them (LCCs can be converted to LDCs with a small loss in parameters).

\begin{define}[$q$-LCCs]
A {\em code} $C \subseteq \Sigma^n$ is called a $q$-query locally correctable code with error-tolerance $\delta >0$, if for every $i \in [n]$ there is a family (called a {\em $q$-matching}) $M_i$, of at least $\delta n$ {\em disjoint} $q$-subsets of $[n]$, with the following {\em decodability} property. For every codeword $c\in C$, and for every $i \in [n]$, the value of $c_i$  is {\em determined}\footnote{Through some function that does not depend on the codeword $c$.} by the values of $c$ in coordinates $S$, for every $q$-subset $S$ in $M_i$.\footnote{ Our definition is a `zero-error' version of the standard definition. By `zero-error' we mean that for any codeword $c$, the value of $c_i$ can be determined correctly (without error) from the coordinates of $c$ at any $q$-subset $S$ in the matching $M_i$. A more general definition would say that $c_i$ can be computed from $c|_S$ with high probability, or even just slightly better than a random guess. Our definition is equivalent to the more general definition for linear codes, which comprise all of the interesting examples. We still allow `global' error in the sense that a large (constant) fraction of the coordinates can be corrupted (this global error tolerance is captured by the parameter $\delta$).} 
\end{define}

Intuitively, given a vector $c' \in \Sigma^n$ which results from corrupting less than (say) $\epsilon \delta n$ coordinates of a codeword $c\in C$, recovering $c_i$ for any given $i\in [n]$ is simple. Picking a random $q$-subset from $M_i$ and decoding $c_i$ according to it will succeed with probability at least $1-\eps$, as only an $\eps$-fraction of these $q$-subsets can be corrupted.

We focus in this paper on the most well-studied and well-motivated regime where both the ``query-complexity'' $q$ and the error-tolerance $\delta$ are constants. It is not hard to see that there are no LCCs with $q=1$  (unless the dimension is constant) and so we will start with the first interesting case of $q=2$.
A canonical example of  a 2-query LCC, which will serve us several times below, is the Hadamard code. Here $\Sigma = \F_2$. Let $k$ be any integer and set $n=2^k-1$. Let $A$ be the $k\times n$ matrix whose columns are all non-zero $k$-bit vectors. The Hadamard code $C_H \in \F_2^n$ is {\em generated} by $A$, namely $H$ consists of all linear combinations of rows of $A$. Since every column of $A$ can be written as a sum of (namely, {\em spanned} by) pairs of other columns in $(n-1)/2$ different ways, the matching $M_i$ suggest themselves, and so is the {\em linear} correcting procedure:  add the values in coordinates of the random pair $S$  from $M_i$ to determine the $i$th coordinate.

A central parameter of codes is their {\em rate}, capturing the redundancy between the {\em dimension}, namely the number of information bits encoded (here $k$), and the length of the codeword (here $n$). As in this paper this $k$ will be a tiny function of $n$, we will focus on the {\em dimension} itself. Note that in the example above, as in every {\em linear} code, this dimension is also the {\em rank} of the generating matrix.
In general codes, dimension may be fractional, and is defined as follows. All logarithms are in base 2 unless otherwise noted.

\begin{define}[Dimension and rate of a code]
For a general, possibly non-linear code $C\subseteq \Sigma^n$, we define the dimension of $C$  to be $ \dim(C) = \log |C| / \log |\Sigma|$. Note that this coincides with the linear algebraic definition of dimension when $C$ is a subspace. We refer to the ratio $\dim(C)/n$ as the `rate' of the code.
\end{define}

Note that while the Hadamard code ($C_H$) has fantastic local correction (only 2 queries), its  dimension is only  $k \sim \log n$, which is pathetic from a coding theory perspective. However, no better 2-query LCC can exist, regardless of the alphabet. 


\begin{thm}[2-LCCs]\label{2-LCCs}
For all large enough $n$ and over any alphabet:
\begin{itemize}
\item There exists a 2-query LCC of dimension $\Omega (\log n)$ and constant $\delta$ (Folklore: Hadamard code). 
\item Every 2-query LCC must have dimension at most $O(\log n)$ (for any constant $\delta$) \cite{BhattacharyyaGT17}. 
\end{itemize}
\end{thm}

While we know precisely the optimal dimension for 2 queries, for $q\geq 3$ the gap between known upper and lower bounds is huge. The best lower bounds (constructions) are polylogarithmic: they come from Reed-Muller codes (using polynomials over finite fields), and yield  dimension $\Omega((\log n)^{q-1})$.

The best LCC upper bounds are only slightly sub-linear, giving  $\dim(C) \leq \widetilde{O}(n^{1-\frac{1}{q-1}})$ (up to logarithmic factors). This bound, which we will refer to as the Katz-Trevisan (KT) bound, is actually a slight refinement/improvement over the bound originally appearing in \cite{KatzT00} (which gave $n^{1 - 1/q}$). This improvement was implicit in several works (e.g. \cite{DinurK11,Woodruff07}) and is explicitly stated in \cite{IcelandS18}. We should also note that, over constant-size alphabets, Kerenidis and De-Wolf proved an even stronger bound using quantum information theory~\cite{KerenidisW04}. This exponential gap between the upper and lower bounds, which we formally state below, has not been narrowed in over two decades.\footnote{For LDCs better constructions than Reed-Muller codes are known, through the seminal works of \cite{Yekhanin08,Efremenko09}, but as mentioned we will not discuss them here. Still, the upper bounds for LDCs are the same as for LCCs, and obtained by the same KT-type argument, so the results in this paper may serve to better understand the (smaller, but still quite large) gap between upper and lower bounds in LDCs as well.} Explaining  this gap  (in the hope of finding ways to close it) is one major motivation of this work.

\begin{thm}[$q$-LCCs, $q\geq 3$]\label{3-LCCs}
For every fixed $q\geq 3$ and all large enough $n$:
\begin{itemize}
\item There exists a $q$-query LCC of dimension $\Omega((\log n)^{q-1})$ (with constant $\delta$ and alphabet of size $q+1$) (Reed-Muller codes, see e.g. the survey \cite{Yekhanin12}).
\item Every $q$-query LCC must have dimension at most $\widetilde{O}(n^{1-\frac{1}{q-1}})$ (for any constant $\delta$ and any alphabet) \cite{IcelandS18}.
\end{itemize}
\end{thm}

\subsection{Spanoids} 
  
We shall now abstract the notion of inference used in LCCs. There, for a collection of pairs $(S,i)$ with $S\subseteq [n]$ and $i\in [n]$, the values of codewords in coordinate positions $S$, determine the value of of some other coordinate $i$. We shall forget (for now) the underlying code altogether, and abstract this relation by the formal ``inference'' symbol  $S\rightarrow i$, to be read ``$S$ {\em spans} $i$''. 

\begin{define}[Spanoid]
\label{def-spanoid}
A {\em spanoid} $\cS$ over $[n]$ is a family of pairs $(S,i)$ with $S\subseteq [n]$ and $i\in [n]$. The pair $(S,i)$ will sometimes be written as $S \rightarrow i$ and read as $S$ spans $i$ in the spanoid $\cS$.  
\end{define}


One natural way to view a spanoid is as a logical inference system, with the pairs indicating all inference rules. The elements of $[n]$ indicate some $n$ formal statements, and an inference $S \rightarrow i$ of the spanoid means that if we know the truth of the statements in $S$, we can infer the truth of the $i$th statement. With this intuition, we shall adopt the convention that the inferences $i \rightarrow i$ are implicit in any spanoid, and that {\em monotonicity} holds: if $S \rightarrow i$ then also $S' \rightarrow i$ for every $S' \supseteq S$. These conventions will be formally stated below when we define general  derivations, which  sequentially combine these implicit rules and the stated rules (pairs) of the spanoid.

A key concept of spanoids is, naturally, the {\em span}.
Given a subset $T\subset [n]$ (which we can think of as ``axioms''), we can explore everything they can span by a sequence of applications of the inference rules of the spanoid $\cS$. 

\begin{define}[Derivation, Span]
A {\em derivation} in $\cS$ of $i\in [n]$ from $T\subseteq [n]$, written $T \models_{\cS} i$, is a sequence of sets $T=T_0, T_1,\dots ,T_r$ with $i \in T_r$ such that for each $j\in [r]$,  $T_{j} = T_{j-1}\cup {i_j}$ for some $i_j \in [n]$ and there exists $S \subset T_{j-1}$ such that $(S,i_j) \in \cS$ is one of the spanoid rules.

 The {\em span} (or {\em closure}) of $T$, denoted $\spn_{\cS} (T)$, is the set of all $i$ for which $T \models_{\cS} i$. We shall remove the subscript $\cS$ from these notations when no confusion about the underlying spanoid can arise, and write $T \models  i$ and $\spn(T)$ for short.
\end{define}

Despite being highly abstract, we will see that spanoids can lead to a rich family of questions and definitions. The first, and perhaps one of the most central definitions  is that of the {\em rank} of a spanoid. We shall see other notions of spanoid rank later on (and will discuss the relation between them).

\begin{define}[Rank]
The {\em rank} of a spanoid $\cS$, denoted $\rank(\cS)$, is the size of the {\em smallest} subset $T\subseteq [n]$ such that $\spn(T)=[n]$. Note that by the definition of span we always have $\rank(\cS) \leq n$. 
\end{define}

We  note that the ``rank'' of a logical inference system does appear (under different names) in {\em proof complexity}. It is the starting point for {\em expansion-based} lower bounds on   a variety of proof systems, as introduced for Resolution proofs in \cite{BenSassonW01}, and used for many others e.g. in   \cite{AlekhnovichBRW04} and  \cite{AlekhnovichR01}). We shall return to this connection presently.

We can now define the spanoid analog of $q$-LCCs as spanoids which only specify the correction structure (the matchings $M_i$) without requiring any codewords or alphabet.

\begin{define}[$q$-LCS, Locally correctable spanoid]
A spanoid $\cS$ over $[n]$ is a $q$-LCS with error-tolerance $\delta$ if	for every $i \in [n]$ there exists a family  $M_i$ of at least $\delta n$ disjoint $q$-subsets of $[n]$ such that for each $S \in M_i$ we have $(S,i) \in \cS$. Namely, each $i \in [n]$ is spanned (in $\cS$) by at least $\delta n$ disjoint  subsets of $q$-elements.
\end{define}

One can now ask about the highest possible rank of a $q$-LCS. It is not hard to see that the existence of a $q$-LCC (over any alphabet) $C \subset \Sigma^n$ with dimension $\dim(C)=d$ automatically implies that there exists a $q$-LCS (namely, the one given by the same matchings used in $C$) with rank at least $\lceil d \rceil$. Indeed, otherwise  there would be $r < d$ coordinates  in $[n]$ that determine any codeword $c \in C$ and this would limit the number of codewords to $\Sigma^r$.

One of our main observations is that, remarkably, in locally correctable spanoids there is {\em no} gap between the upper and lower bounds: we know  the precise answer up to logarithmic factors, and it matches the {\em upper bounds} for LCCs!  Observe the analogies to the theorems in the previous subsection, for $q=2$ and $q\geq 3$.

\begin{thm}[$2$-LCSs]\label{thm-2-spanoids}
For all large enough $n$:
\begin{itemize}
\item There exists a 2-LCS over $[n]$ with error-tolerance $\delta$ of rank $\Omega(\frac{1}{\delta}\log (\delta n))$.
\item Every  2-LCS  over $[n]$ with error-tolerance $\delta$ must have rank at most $O\left(\frac{1}{\delta}\log(n)\right).$\footnote{The results of \cite{BhattacharyyaGT17} can be interpreted as an upper bound of $O\left(\poly(1/\delta)\log(n)\right)$ on the rank of 2-LCS with error-tolerance $\delta$.}
\end{itemize}
\end{thm}

  
Here, of course, the inference structure of the Hadamard code proves the first item. To get the required dependence on $\delta$, one can take $\frac{1}{\delta}$ disjoint copies of such spanoids. The second item requires a new proof we discuss below, which generalizes (and implies) the one in Theorem~\ref{2-LCCs}. It is quite surprising that, even in this abstract setting, with no need for codewords or alphabet,  one cannot do better than the Hadamard code!

We now state our results for $q \geq 3$. 

\begin{thm}[$q$-LCSs with $q\geq 3$]\label{thm-q-spanoids}
For every fixed $q\geq 3$ and all large enough $n$:
\begin{itemize}
\item There exist a  $q$-LCS of rank $ \tilde\Omega(n^{1 - \frac{1}{q-1}})$ (with constant $\delta$).
\item Every $q$-LCS over $[n]$  has rank at most  $\tilde O\left(n^{1 - \frac{1}{q-1}} \right)$ (for any constant $\delta$).
\end{itemize}
\end{thm}
  
Both parts of this theorem demand discussion. 
The possibly surprising (and tight) lower bound follows from a simple probabilistic argument  (indeed, one which is repeatedly used to prove expansion in the proof complexity references cited above), where the matchings $M_i$ are simply chosen uniformly at random.  It seems to reveal how significant a relaxation spanoids are of LCCs (where probabilistic arguments fail completely). 
However, the best known  LCC upper bound (Theorem~\ref{3-LCCs}) does not rule  out the possibility that, at least for large alphabets, the two (LCC's dimension and LCS's rank) have the same behavior!  From a more pessimistic  (and perhaps more realistic) perspective, our lower bound shows the limitations of any (upper bound) proof technique which, in effect, applies also for spanoids. These are proofs in which the LCC structure is used to show that a small subset spans all the others. We note that there are several LCC upper bounds which `beat' the $n^{1 - \frac{1}{q-1}}$ bound for certain very special cases by using additional structure not present in the corresponding abstract spanoid. One example is the bound of \cite{KerenidisW04}, which uses arguments from quantum information theory to roughly {\em halve} the number of queries, over binary (or small) alphabets. Another example is the paper \cite{DvirSW14}, which gives an improved upper bound on the dimension for linear $3$-LCCs defined over the real numbers, using specific properties of the Reals such as distance and volume arguments.

Our proof of the upper bound, is again more general than for LCCs, and  interesting in its own right. We use a simple technique which performs random restrictions and contractions of graphs and hypergraphs (and originates in~\cite{DvirSW14}). It will be described in Section~\ref{sec-upperbounds}, after we have formulated an equivalent, set-theoretic formulation of spanoids in Section~\ref{sec-sets}.

\subsubsection{Functional rank: bridging the gap between LCCs and LCSs}

We conclude this section of the introduction with an attempt to understand (and possibly bridge) the gap between LCCs and their spanoid abstraction. The idea is to start with an LCS of high rank (which we know is possible), and convert it to an LCC without losing too much in the parameters. More generally, for  a given spanoid $\cS$, we would like to investigate the  code $C$ with largest dimension (over any alphabet $\Sigma$) which would be consistent with the inferences of $\cS$. This is captured in the notion of  {\em functional rank} which we now define.

\begin{define}[Functional rank]
Let  $\cS$ be a spanoid over $[n]$.  A code $C\subset \Sigma^n$ is {\em consistent} with $\cS$ if for every inference $(S,i)$ in $\cS$, and for every codeword $c\in C$, its values of coordinates $S$ determine  its value in coordinate $i$ (by some fixed function, $f_{S,i}$ not depending on $c$).\footnote{One can think of a code consistent with $\cS$ also as a `representation' of $\cS$ in the spirit of matroid theory.} 

Define the {\em functional rank} of $\cS$, denoted $\frank(\cS)$, to be equal to the supremum of the dimension $\dim(C)$, over all possible finite alphabets $\Sigma$ and codes  $C \subset \Sigma^n$ which are consistent with $\cS$.
\end{define}

Of course, the strategy of constructing LCCs in two stages as above can only work if we can bound the gap between $\rank(\cS)$ and $\frank(\cS)$. This question, of bounding this gap or proving it can be large, is perhaps the most interesting one we raise (and leave mostly open for now). For now, we are able to show an example in which the two are different. The example providing a gap is depicted in Figure~\ref{fig-pentagon}, arranging the coordinates as the vertices of a pentagon, the pair of vertices of each edge span the vertex opposite to it. That is, $\{x_1,x_2\} \rightarrow  x_4, \{x_2,x_3\} \rightarrow x_5$ etc.

\begin{figure}[!h]
\caption{The pentagon spanoid $\Pi_5$ where each coordinate is spanned by the coordinates of the opposite edge.}
\centering
\definecolor{ffqqqq}{rgb}{1,0,0}
\definecolor{rvwvcq}{rgb}{0.08235294117647059,0.396078431372549,0.7529411764705882}
\begin{tikzpicture}
\label{fig-pentagon}
\draw [line width=2pt,color=rvwvcq] (4.16,-2.8)-- (6.22,-2.78);
\draw [line width=2pt,color=rvwvcq] (6.22,-2.78)-- (6.837553878086487,-0.8146432365444852);
\draw [line width=2pt,color=rvwvcq] (6.837553878086487,-0.8146432365444852)-- (5.159223164628246,0.38001404329050925);
\draw [line width=2pt,color=rvwvcq] (5.159223164628246,0.38001404329050925)-- (3.5044038612617054,-0.8470039163194834);
\draw [line width=2pt,color=rvwvcq] (3.5044038612617054,-0.8470039163194834)-- (4.16,-2.8);
\draw [->,line width=1pt,dash pattern=on 2pt off 2pt,color=ffqqqq] (5.26783713739075,-2.3236488686072088) -- (5.16,0);
\draw [->,line width=1pt,dash pattern=on 2pt off 2pt,color=ffqqqq] (4.348459793149962,-1.6691768608425783) -- (6.4,-0.96);
\draw [->,line width=1pt,dash pattern=on 2pt off 2pt,color=ffqqqq] (4.551034462219967,-0.5939728480863998) -- (6,-2.32);
\draw [->,line width=1pt,dash pattern=on 2pt off 2pt,color=ffqqqq] (5.797647810343071,-0.5160595138287057) -- (4.54,-2.36);
\draw [->,line width=1pt,dash pattern=on 2pt off 2pt,color=ffqqqq] (6.109301147373848,-1.6535941939910395) -- (4.02,-1.02);
\begin{scriptsize}
\draw [fill=black] (4.16,-2.8) circle (2.5pt);
\draw[color=black] (3.832877126298422,-2.924936203142211) node {$x_1$};
\draw [fill=black] (6.22,-2.78) circle (2.5pt);
\draw[color=black] (6.589789150701544,-2.924936203142211) node {$x_2$};
\draw [fill=black] (6.837553878086487,-0.8146432365444852) circle (2.5pt);
\draw[color=black] (7.153095824763097,-0.6303548461452422) node {$x_3$};
\draw [fill=black] (5.159223164628246,0.38001404329050925) circle (2.5pt);
\draw[color=black] (5.130167804796906,0.7695105014232468) node {$x_4$};
\draw [fill=black] (3.5044038612617054,-0.8470039163194834) circle (2.5pt);
\draw[color=black] (3.1484051185337923,-0.66152017984831994) node {$x_5$};
\end{scriptsize}
\end{tikzpicture}
\end{figure}
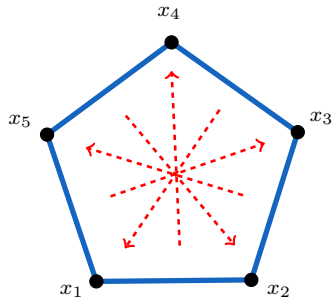

\begin{thm}[Constant gap between rank and functional rank]
The pentagon spanoid $\Pi_5$ depicted in Figure~\ref{fig-pentagon} has $\rank(\Pi_5)=3$ but $\frank(\Pi_5)=2.5$.
\end{thm}

Seeing that $\rank(\Pi_5)=3$ is easy by inspection. The lower bound of 2.5 on functional rank comes from a set-theoretic construction of consistent codes we will describe (in more generality) in Section~\ref{sec-LPcover}, where we develop an linear programming (LP) relaxation for $\rank(\cS)$ called $\LPcover(\cS)$, but surprisingly this LP captures the best set-theoretic construction of consistent codes. But even in this small example, the upper bound on functional rank is nontrivial to determine, as we allow all possible alphabets and consistent codes. Not surprisingly, Shannon entropy is the key to proving such a bound. In Section~\ref{sec-LPentropy}, we develop a linear programming relaxation, based on entropy whose optimum $\LPentropy(\cS)$ {\em upper bounds} $\frank(\cS)$. In this example, it proves 2.5 to be the optimum.

One natural way of amplifying gaps as in the example above, which may also be useful in creating codes of high functional rank, is the idea of {\em tensoring}. In Section~\ref{sec-spanoidproduct}, we develop different notions of tensoring {\em spanoids} inspired by tensoring of codes. In particular, we define a product of spanoids called the {\em semi-direct product} under which $\rank$ is multiplicative and $\frank$ is sub-multiplicative. By repeatedly applying this product to $\Pi_5$, we get a spanoid with polynomial gap between $\frank$ and $\rank$.

\begin{thm}[Polynomial gap between rank and function rank]\label{thm-FrkvsRk}
  There exists a spanoid $\cS$ on $n$ elements with $\rank(\cS) \ge n^c \frank(\cS)$
  where $c = \log_5 3 - \log_5 2.5 \ge 0.113$.
  \end{thm}


Summarizing, we have the following obvious inequalities between the measures we described so far for every spanoid $\cS$. We feel  that understanding the exact relationships better is worthy of further study 
\begin{equation}
\label{eqn-ranks-order}
\LPcover(\cS) \leq \frank(\cS) \leq \LPentropy(\cS) \leq \rank(\cS).
\end{equation}

\subsection{Other motivations and incarnations of spanoids} 

We return to discuss other structures, combinatorial, geometric and algebraic, in which the same notions of span and inference naturally occur, leading to a set-theoretic one that elegantly captures spanoids precisely. These raise further issues, some of which we study in this paper and some are left for future work. These serves to illustrate the breadth of the spanoid framework.

\subsubsection{Bootstrap percolation and gossip processes}

The following general set-up occurs in statistical physics, network theory and probability theory. 
Fix an undirected graph $G([n],E)$. In a gossip or infection process, or equivalently bootstrap percolation, we are give a set of ``rules'' specifying, for every vertex $v\in [n]$,  a family of subsets of its neighbors. The intended meaning of such a rule is  that if {\em every} member of one such subset is ``infected'' at a certain time step, then the vertex $v$ becomes infected in the next time step. Given a  set of initial infected vertices, this defines a process in which infection spreads, and eventually stabilizes. A well studied special case is  the (uniform) {\em $r$-bond percolation}~\cite{Bollobas68}, where the family for each vertex is  all $r$-subsets of its neighbors. Many variants exist, e.g. one can have a similar process on the edges, rather than vertices of the graph. An important parameter of such a process is the following: what is the size of the smallest set of vertices which, if infected, will eventually infect all other vertices.\footnote{This turns out to be crucial for understanding, at least for certain structured graphs like lattices studied by physicists, the threshold probability for percolation when initial infections are random.}

A moment's thought will convince the reader that this structure is precisely a spanoid (where inferring sets are restricted by the graph structure). The infection process is precisely the inference process defining span in spanoids. Furthermore, the smallest size of an infecting set is {\em precisely} the rank of that spanoid! Much work has been invested to determine that rank even in very special cases, e.g. for the $r$-bond percolation above, in e.g. Boolean hypercubes, where it is known precisely. Interestingly, the paper~\cite{HambardzumyanHY17} uses the so-called ``polynomial method'' to reprove that bound, which fits even deeper with our framework. In our language, their method determines the {\em functional rank} of this spanoid, and one direction is through constructing an explicit code that is consistent with the spanoid! The reader is encouraged to work out the details. 

\subsubsection{Independence systems and Matroids}
An {\em independence system} over $[n]$ is a family $\cF$ of subsets of $[n]$ which is downwards-closed (if a set is in $\cF$, so are all its subsets). The members of $\cF$ are called {\em independent}. While much of what we say below generalizes to all independence systems, we specify them for the important special systems called {\em matroids}.

A {\em matroid} is an independence system in which the independent sets satisfy the so-called ``exchange axiom'' (which we will not define here). Matroids abstract {\em linear} independence in subsets of a vector space over a field\footnote{Matroids are in fact more general than linear independent sets of vectors over a field, for example the V\'amos matroid on eight elements is not representable over any field.}, and capture algorithmic problems in which optimization is possible through the greedy algorithm. Matroids thus come with natural notions of span and rank, extending the ones in the linear algebraic setting. The rank of a set is the size of the largest independent set it contains. The span of a set is the maximal superset of it of the same rank.
 A matroid can thus be naturally viewed as a spanoid, with the inference rules $F \rightarrow i$ for every independent  $F \in \cF$ and every $i$ for which $F \cup \{i\}$ is {\em not} independent (such minimal dependent sets as $F  \cup \{i\}$ are called {\em cycles}). It is easy to verify that the notions of span and rank of the matroid and the spanoid it defines coincide. This also raises the natural question of bounding the gap between $\frank$ and $\rank$ for the special case of spanoids arising from matroids.

Note that a spanoid resulting from a matroid this way is {\em symmetric}: by the exchange property of matroids, if $E\subset [n]$ is a cycle of $\cF$, then for {\em every} $i\in E$ it contains the inference $E\setminus \{i\} \rightarrow i$. Symmetric spanoids are  interesting, and we note that  the pentagon example witnessing the gap between rank and functional rank is {\em not} symmetric, and we do not know such a gap for symmetric spanoids. We also don't know if symmetric spanoids can achieve the lower bound in Theorem~\ref{thm-q-spanoids}.

\subsubsection{Point-line incidences} 
Sylvester-Gallai theorem is a celebrated result in combinatorial geometry conjectured by Sylvester and proved independently by Melchior and Gallai. It states that for any  set of $n$ points in Euclidean space $\R^d$, if the line through any two points passes through a third point, then they must all be collinear (namely, they span  a 1-dimensional affine space). Over the complex numbers, one can prove a similar theorem but with the conclusion that the points span a 2-dimensional affine space (and there are in fact two dimensional examples known) \cite{Kelly86}. Over finite fields the conclusion is even weaker, saying that the span has dimension at most $O(\log(n))$ and this is tight as the example of all points in $\F_p^k$ with $n=p^k$ shows.  It is not a coincidence that this example reminds one of the Hadamard code described before as an example of a $2$-query LCC. It is in fact true that there is a tight connection between configurations of points with many collinear triples and linear $2$-query LCCs. This  was first noticed in \cite{BarakDWY11,BDWY12} and was used to prove that $2$-LCCs do not exist over the characteristic zero fields (for $q\geq 3$ these questions are wide open with even larger gaps than in the finite field case). LCCs with more than 2 queries naturally correspond to point configurations with many $(q-2)$-dimensional affine spaces containing at least $q$ points.

Given the connections between Sylvester-Gallai type incidence structures and LCCs, and the insights offered by spanoids for studying LCCs, it is natural that the study of the spanoid structures can help us get new insights on incidence geometry problems. 
A {\em dual}  way to view these incidences, which we shall presently generalize, is to consider each point $p_i \,:\,i\in [n]$ as representing a hyperplane $F_i$ through the origin (in the appropriate vector space) vanishing on the linear function defined by $p_i$. A point $p_i$ is spanned by a collection of points $\{p_j: j\in S\}$ iff $F_i\supset \cap_{j\in S}F_j$. Therefore the spanning structure of the points $p_1,p_2,\dots,p_n$ is captured by the spanoid where we would add the inference $S\rightarrow i$ iff $F_i$  {\em contains} the common intersection of all $F_j \,: \,j\in S$. 

\subsubsection{Systems of polynomial equations}

Given the above example, there is no reason to stop at the linear setting. Instead of lines we can consider $n$  (multivariate) polynomials $f_i$ over a field, and again consider $F_i$ to be the zero set of $f_i$. The spanoid above, having an inference $S \rightarrow i$ whenever the set $F_i$  {\em contains} the common intersection of all $F_j \,: \,j\in S$, is capturing another natural algebraic notion. Namely, it says that the polynomial $f_i$ vanishes on the all the common roots of the polynomials $\{ f_j \,: \,j\in S\}$. By the celebrated Hilbert's Nullstellensatz theorem, over algebraically closed fields, this implies that $f_i$ belongs to the radical of the ideal generated by the $f_j$'s. Here the rank function is far from being that of a matroid; the complex spanoid which arises (and in general is far from understood) plays a role in arithmetic complexity (a beautiful example is the recent~\cite{Shpilka18} dealing with degree-2 polynomials).

\subsubsection{Intersecting set systems}
Let us remove all restrictions from the origin or nature of the $n$ sets $F_i$ discussed in the previous discussion. Assume we are given any such family $\cF$ of sets (from an arbitrary universe, say $U$). As above, a natural spanoid $\cS_{\cF}$ will have the inference $S \rightarrow i$ whenever the set $F_i$  {\em contains} the common intersection of all $F_j \,: \,j\in S$. Such situations (and hence, spanoids) arise in many questions of extremal set theory, for example the study of (weak) sunflowers, or families with certain forbidden intersection (or union) patterns, 
e.g.~\cite{FoxLS12,ErdosFF85,Furedi96}. 

What is interesting in this more general framework, where the initial family of sets $\cF$ is arbitrary, is that it becomes {\em equivalent} to spanoids! In other words, every spanoid $\cS$ arises as $\cS_{\cF}$ of {\em some} family of sets $\cF$.  This possibly surprising fact is not much more than an observation, but it turns out to be an extremely useful formulation for proving some of the results in this paper. Let us state it formally (it will be proved in Section~\ref{sec-sets}).

\begin{thm}[Spanoids and intersecting sets]
Let $\cS$ be any spanoid on $[n]$. Then, there exists a universe $U$ and a family $\cF$ of $n$ sets of $U$,  $\cF =\{F_1, F_2,\dots ,F_n\}$, such that $\cS = \cS_{\cF}$.
\end{thm}

It is convenient to assume that the sets in $\cF$ have no element in common to all\footnote{Indeed, otherwise we can  remove the common intersection (if non-empty) of {\em all} members of $\cF$ from each of them, as it does not change the underlying spanoid.}.

The notions of rank and span are extremely simple in this set-theoretic setting, and do not require the sequential ``derivation'' and the implicit ordering which we require to define these in spanoids. For a family $\cF$ of $n$ sets and  a subset $S\subset [n]$, let us denote by $\cap S$  
the subset of $U$ which is the intersection of all $\{F_j\,:\, j\in S\}$. Then the rank of $S$ is the size of the smallest subset $S'\subseteq S$ for which $\cap S' = \cap S$. Similarly, the span of $S$ is the largest superset $S'' \supseteq S$ for which $\cap S'' = \cap S$.

These static definitions of rank and span make many things transparent. For example, the expected fact that testing if  the rank of a spanoid (namely the rank of the set $[n]$) is at most some given integer $k$ is $NP-complete$ (Claim~\ref{cla-spanoid-rank-NPcomplete}). Complementing the  sets $F_i$ in $\cF$, and replacing intersection with union, this is precisely the Set Cover problem. This connection also underlies the cover-based linear program discussed earlier, as well as proofs of the main quantitative results Theorem~\ref{thm-2-spanoids} and Theorem~\ref{thm-q-spanoids}.

\subsubsection{Union-closed families}
Spanoids over $[n]$ are equivalent to union-closed families of subsets of $[n]$ i.e. a family of subsets of $[n]$ such that the union of any two members is again in the family. A closed set of a spanoid $\cS$ is a subset $A\subset [n]$ such that $\spn(A)=A$. The family of all closed sets of a spanoid $\cS$ is denoted by $\cC_\cS$ which is an intersection-closed family. Thus the family of all open sets which are complements of closed sets is a union-closed family and is denoted by $\cO_\cS$. One can construct all the derivations of the spanoid given its family of open or closed sets. Conversely, given any union-closed family of subsets of $[n]$, one can define a spanoid whose open sets are precisely the given family. Thus spanoids on $[n]$ are equivalent to union-closed families of subsets of $[n]$. The rank of a spanoid has a very simple interpretation in terms of its open sets, $\rank(\cS)$ is equal to the size of the smallest hitting set for its family of open sets $\cO_\cS$. Moreover, $\rank(\cS)$ is at most $\log|\cO_S|$. These connections are discussed in Section~\ref{sec-closed-open}.

Union-closed families are interesting combinatorial objects with a rich structure. The widely open Frankl's union-closed conjecture states that in every union-closed family of $N$ sets, there exists an element which is contained in at least $N/2$ sets. Though this was proved for various special classes (see survey~\cite{BruhnS15}), the best general bound is $\Omega(N/\log N)$ due to~\cite{Knill94,Wojcik99}. When seen in the framework of spanoids, this follows immediately from Claim~\ref{cla-rankspanoid} which says that there is a $\log N$ sized hitting set for every union-closed family of $N$ sets. Thus there is an element which should hit at least $N/\log_2(N)$ sets. We hope that viewing union-closed families as spanoids could be of use in understanding them.

\subsection{Organization}

In Section~\ref{sec-spanoid-prelim}, we will present two alternative ways to represent spanoids that will be very useful. In Section~\ref{sec-upperbounds}, we show upper bounds on the rank of $q$-LCSs for $q\ge 2$ thus proving the upper bounds in Theorems~\ref{thm-2LCSs} and \ref{thm-qLCSs}. In Section~\ref{sec-construction-qLCS}, we construct $q$-LCSs thus proving the lower bound in Theorem~\ref{thm-qLCSs}. In Section~\ref{sec-frank-vs-rank}, we explore the relation between $\rank$ and $\frank$ of a spanoid. Specifically, we define the linear programs for $\LPentropy$ and $\LPcover$ which provide upper and lower bounds on the $\frank$ respectively. We will also calculate the $\rank$ and $\frank$ of the Pentagon spanoid defined in the introduction (Figure~\ref{fig-pentagon}) and thus proving that $\frank$ can be strictly smaller. In Section~\ref{sec-spanoidproduct}, we study products of spanoids which we use to amplify the gap between $\frank$ and $\rank$ from constant to polynomial.


\section{Preliminaries on spanoids}\label{sec-spanoid-prelim}
We will now describe two equivalent ways to define spanoids which will turn out to be very useful.
\subsection{Spanoids as union-closed or intersection-closed families}
\label{sec-closed-open}
In this subsection, we will define an equivalent and more canonical way of describing spanoids in terms of intersection closed or union closed families. We have defined spanoids in Definition~\ref{def-spanoid} by specifying the initial set of derivation rules such as $A\rightarrow i$. But two different initial set of rules can lead to the same  set of derivations and we should consider two spanoids to be equivalent if they lead to the same set of derivations. We will present an alternative way to describe spanoids which makes them equivalent to union-closed families of sets or alternatively intersection-closed families of sets. Moreover this new representation is a more canonical way to represent spanoids since it will be based only on the set of derivations. For this, the main new notions we need to define are that of a `closed set' and an `open set'.

\begin{define}(Closed and open sets)
\label{def-closed-open-set}
 Let $\cS$ be an spanoid on $[n]$. A {\em closed set}\footnote{Closed sets are analogous to `flats' or `subspaces' in matroids.} is a subset $B \subset [n]$  for which $\spn(B) = B$. A subset $B\subset [n]$ is called an \emph{open set} if its complement is a closed set. The family of all closed sets of $\cS$ is denoted by $\cC_\cS$ and the family of all open sets of $\cS$ by $\cO_\cS$ (when it is clear from the context, we will drop the subscript $\cS$).
\end{define}


\begin{claim}\label{cla-flatsinter}
In any spanoid $\cS$ on $[n]$, 
\begin{enumerate}
 \item the intersection of any number of closed sets is a closed set i.e. $\cC_\cS$ is an intersection-closed family,
 \item the union of any number of open sets is an open set i.e. $\cO_\cS$ is a union-closed family and
 \item for any set $A \subset [n]$, $\spn(A)$ is equal to the intersection of all closed sets containing $A$ i.e. $\spn(A)=\bigcap_{B\supset A, B\in \cC_\cS} B$.
\end{enumerate}
\end{claim}
\begin{proof}
(1) Let $F = F_1 \cap F_2$ be the intersection of two closed sets. Suppose in contradiction that $F$ spans some element $x \in [n] \setminus F$ then, by monotonicity, both $F_1$ and $F_2$ have to span $x$. Hence, $ x \in \spn(F_1) \cap \spn(F_2) = F_1 \cap F_2 = F$ in contradiction.

(2) This just follows from (1) by taking complements.

(3) Let $F(A)$ be the intersection of all closed sets containing $A$. Since $\spn(A)$ is a closed set we clearly have $F(A) \subset \spn(A)$. To see the other direction, suppose $x\in \spn(A)$ and let $F$ be any closed set containing $A$. Then, by monotonicity, $F$ must also span $x$ and so we must have $x \in F$.
\end{proof}


\begin{claim}
\label{cla-spanoid-closedsets}
A spanoid is uniquely determined by the set of all its closed (open) sets which is an intersection-closed (union-closed) family of subsets of $[n]$. Conversely, every intersection-closed (union-closed) family of subsets of $[n]$ defines a spanoid whose closed (open) sets are the given family.
\end{claim}
\begin{proof}
Given a spanoid $\cS$ on $[n]$, by Claim~\ref{cla-flatsinter}, we can define $\spn(A)$ in $\cS$ using just the closed sets as: $$\spn(A)=\bigcap_{B\supset A, B\in \cC_\cS}B.$$
And $A\spans i$ in $\cS$ iff $i\in \spn(A)$. Thus given the set of all closed sets, we can reconstruct all the derivations of the spanoid.

For the converse, suppose we are given an intersection-closed family of subsets of $[n]$, say $\cC$. We can define $\spn_\cC(A)=\bigcap_{B\supset A, B\in \cC}B$ and define a spanoid $\cS_\cC$ where $A\spans i$ iff $i\in \spn_\cC(A)$. It is easy to see that the closed sets of this spanoid $\cS_\cC$ is exactly $\cC$.
\end{proof}

Thus an equivalent way to define a spanoid is to define all its closed (open) sets which is some intersection (union) closed family. The following claim shows that the rank of a spanoid has a very natural interpretation in terms of the open sets. 

\begin{claim}
\label{cla-rank-hittingset}
The rank of a spanoid $\cS$ is the size of the smallest hitting set for the collection $\cO_\cS$ i.e. a set which intersects every open set in $\cO_\cS$ non-trivially.
\end{claim}
\begin{proof}
 Observe that a subset $A\subset [n]$ spans $[n]$ iff it is a hitting set for all the open sets in $\cO_\cS$. This is because if $A$ doesn't hit some open set B, then $A$ lies in the complement of $B$ i.e. $A\subset \bar{B}$. Since $\bar{B}$ is closed, $\spn(A)\subset \bar{B} \ne [n]$. Therefore $\rank(\cS)$ is the size of the smallest hitting set for $\cO_\cS$.
\end{proof}
This interpretation of the rank is used in Section~\ref{sec-LPcover} to give a linear programming relaxation $\LPcover$ which lower bounds the rank. We can also upper bound the rank of a spanoid in terms of the number of closed or open sets as the following claim shows.


\begin{claim}\label{cla-rankspanoid}
Let $\cS$ be a spanoid, then $\rank(\cS)\le \log_2(|\cC_\cS|)=\log_2(|\cO_\cS|)$.
\end{claim}
\begin{proof}
Let $r=\rank(\cS)$ and $R \subset [n]$ be a  set of size $|R|=r$ spanning $[n]$. Since the rank of $\cS$ is $r$ we know that $R$ is independent (not spanned by any proper subset). For each of the $2^r$ subsets  $S \in 2^R$ we consider the closed set $F_S = \spn(S)$. We claim that all of these are distinct. Suppose in contradiction that there were two distinct sets $S \neq T \in 2^R$ with $\spn(S) = \spn(T)$. W.l.o.g suppose there is an element $x \in T \setminus S$. Then $x \in \spn(S)$ and so we get that $R \setminus \{x\}$ spans $R$ (by monotonicity) and so spans the entire spanoid in contradiction. Thus $|\cC_\cS|\ge 2^r$.
\end{proof}

\subsection{Spanoids as set systems}\label{sec-sets}
In this subsection, we will show yet another way of representing spanoids by families of sets. This representation (which is equivalent to spanoids) will be easier to work with and, in fact, we will later work almost exclusively with it instead of with the definition given in the introduction.
Recall the notation introduced at the end of the introduction that, for sets $S_1,\ldots,S_n$ and for a subset $A \subset [n]$ we let $\cap A = \cap_{i \in A} S_i$.

\begin{define}[Intersection Dimension of a set system]
The {\em intersection-dimension} of a family of sets $S_1,\ldots,S_n$, denoted $\idim(S_1,\ldots,S_n)$ is the smallest integer $d$ such that there exist a set $A \subset [n]$ of size $d$ such that  $\cap A = \cap [n]$. 
\end{define}

\begin{lem}[Set-Representation of spanoids]\label{lem-represent}
Let $\cS$ be a spanoid on $[n] $ with   $\rank(\cS) = r$. Then there exists a family of sets $S_1,\ldots, S_n$    such that  $A \spans i$  in $\cS$ iff    $\cap A \subset S_i$. In this case we say that the set family $(S_1,S_2,\dots,S_n)$ is a {\em set-representation} of $\cS$ and this implies in particular that $\idim(S_1,\ldots,S_n) = \rank(\cS)$.
\end{lem}
\begin{proof}
For $i \in [n]$ we define $S_i \subset \cC_\cS$ to be the subfamily of closed sets of $\cS$ containing the element	$i \in [n]$. For the first direction of the proof suppose that $A$ spans $x$ in the spanoid $\cS$. Then, by Claim~\ref{cla-flatsinter}, $x$ belongs to any closed set containing $A$ and so  $\cap_{i\in A} S_i \subset S_x$. For the other direction, suppose $\cap_{i\in A} S_i \subset S_x$ or that any closed set containing $A$ must also contain $x$. Hence, $x$ is in the intersection of all closed sets containing $A$ and, by Claim~\ref{cla-flatsinter} we have that $x \in \spn(A)$.  
\end{proof}

An alternative way to represent spanoids is by unions. 
$(T_1,T_2,\dots,T_n)$ is called a \emph{union set-representation} of the spanoid $\cS$ when, $A\spans i$ in $\cS$ iff $T_i\subset \cup_{j\in A}T_j$. Note that if $(S_1,S_2,\dots,S_n)$ is an (intersection) set-representation for $\cS$ as in Lemma~\ref{lem-represent}, then by taking complements, $(\bar{S}_1,\bar{S}_2,\dots,\bar{S}_n)$ is a union set-representation for $\cS$ and vice versa. Thus these two notions of representing a spanoid by sets is equivalent.

\begin{claim}\label{cla-spanoid-rank-NPcomplete}
Given a spanoid $S$ and some positive integer $k$, deciding if the rank of the spanoid is at most $k$ is NP-complete.
\end{claim}
\begin{proof}
Given the description of a spanoid and a subset of its elements, we can check in polynomial time whether the subset has size at most $k$ and spans all the elements. So the problem is in $NP$. To prove that it is $NP-complete$, we reduce Set Cover problem to this. 

Given a collection of sets $S_1,S_2,\dots,S_n\subset U$ such that $\cup_i S_i = U$ and some positive integer $k$, the Set Cover problem asks if there are at most $k$ sets in the collection whose union is $U$. To reduce it to the spanoid rank problem, we can create a spanoid over $[n]$ elements where the inference rules are given by $A\spans i$ iff $\cup_{j\in A} S_j \supset S_i$. The rank of this spanoid is at most $k$ iff there exists $k$ sets in the collection which cover all of $U$. 
\end{proof}

\section{Upper bounds on the rank of $q$-LCSs}\label{sec-upperbounds}

In this section we prove the upper bounds on the rank of $q$-LCSs stated in Theorems \ref{thm-2-spanoids} and \ref{thm-q-spanoids}. The proofs will rely on the set representation described in Section~\ref{sec-sets} and on random restriction and contraction arguments given below.

\subsection{Graph theoretic lemmas}
In this subsection, we will prove a key technical lemma about a random graph process that will be useful for proving upper bounds on the rank of $q$-LCSs. We denote by $\DG(n)$ the set of simple directed graphs on $n$ vertices. We always assume w.l.o.g that the set of vertices are the integers between $1$ and $n$.
 
 \begin{define}[$(\alpha,\beta)$-spread distribution]\label{def-spread}
 Let $\mu$ be a distribution on $\DG(n)$. We say that $\mu$ is {\em $(\alpha,\beta)$-spread} if the following conditions are true for a graph $G$ sampled from $\mu$:
 \begin{enumerate}
 \item Each vertex $i \in [n]$ has an incoming edge with probability at least $\alpha$ i.e. $$\forall i\ \prob_{G\sim \mu}\left[\exists j: (j,i)\in E(G)\right] \ge \alpha.$$
 \item For every $i,j\in [n]$, the probability that $(j,i)$ is an edge is at most $\beta/n$ i.e. $$\forall i,j\ \prob_{G\sim \mu}\left[(j,i)\in E(G)\right] \le \frac{\beta}{n}.$$
\end{enumerate}  
\end{define}
 
For example, one can generate an $(k/n,1)$-spread distribution $\mu$ on $\DG(n)$ in the following way: Fix arbitrary sets $S_1,\ldots,S_n \subset [n]$ of size $k$ each. To sample a graph $G$ from $\mu$, pick a uniformly random element $j \in [n]$ and let $G$ be the directed graph containing the edges $(j,i)$ for each $i$ such that $j \in S_i$. This satisfies the definition since for any fixed $i\in [n]$, $i$ has an incoming edge if $j\in S_i$ which happens with probability $|S_i|/n = k/n$. And for any fixed $i',j'\in [n]$, the probability that $(j',i')$ is an edge is at most $1/n$ since this happens only when $j'=j$ and $j$ is chosen uniformly at random from $[n]$. Note that the sampled edges overall are highly correlated (they all have $j$ as an endpoint).

\noindent We will need following simple observation about $(\alpha,\beta)$-spread distributions.
\begin{lem}\label{lem-alphabetaspread}
Let $\mu$ be an $(\alpha,\beta)$-spread distribution on $\DG(n)$. For every vertex $i$ and every subset $S\subset [n]$ of size at most $\frac{\alpha n}{2\beta}$,
$$\prob_{G\sim \mu}[\exists j\notin S: (j,i)\in E(G)]\ge \frac{\alpha}{2}.$$
\end{lem}
\begin{proof}
This follows from union bound and properties of $(\alpha,\beta)$-spread distributions.
\begin{align*}
\alpha &\le \prob_{G\sim \mu}[\exists j: (j,i)\in E(G)]\\
& \le \prob[\exists j\in S: (j,i)\in E(G)]+\prob[\exists j\notin S: (j,i)\in E(G)]\\
&\le \sum_{j\in S} \prob[(j,i)\in E(G)] + \prob[\exists j\notin S: (j,i)\in E(G)]\\
&\le \frac{\alpha n}{2\beta}\cdot \frac{\beta}{n} + \prob[\exists j\notin S: (j,i)\in E(G)]\\
& = \frac{\alpha}{2} + \prob[\exists j\notin S: (j,i)\in E(G)]
\end{align*}
\end{proof}

Given a distribution $\mu$ on graphs we would like to study the random process in which we, at each iteration, sample from $\mu$ and `add' the edges we got to the graph obtained so far. For two graphs $G$ and $H$ on the same set of vertices, we denote by $G \cup H$ their set theoretic union (as a union of edges).

\begin{define}[Graph process associated with $\mu$]
Let $\mu$ be a distribution on $\DG(n)$. We define a sequence of random variables $G^\mu_t$, $t = 0,1,2,\ldots$ as follows. $G^\mu_0$ is the empty graph on $[n]$ vertices. At each step $t\geq 1$ we sample a graph $G$ according to $\mu$ (independently from all previous samples) and set $G_t = G_{t-1} \cup G$.
\end{define}

For a graph $G \in \DG(n)$ and a vertex $i \in [n]$ we denote by $\reach(i)$ the set of vertices that are reachable from $i$ (via walking on directed edges). By convention, a vertex is always reachable from itself. Similarly, for a set of vertices $S \subset [n]$ we denote by $\reach(S) = \cup_{i \in S}\reach(i)$ the set of vertices reachable from some vertex in $S$. We denote the set of strongly connected components of $G$ by $\Gamma(G)$. We  denote by $C(i) \in \Gamma(G)$ the strongly connected component of $G$ containing $i$. We say that $C \in \Gamma(G)$ is a {\em source}  if $C$ has no incoming edges from any vertex not in $C$. 

\begin{lem}\label{lem-directed}
Let $\mu$ be an $(\alpha,\beta)$-spread distribution on $\DG(n)$ and let $G_t^\mu$ be its associated graph process. Then, for all $t \geq 0$, there is positive probability that the graph  $\Gamma(G_t^\mu)$ has at most $$ n \cdot (1 - \alpha/4)^t + \frac{2\beta}{\alpha}$$ sources. 
\end{lem}
\begin{proof}
If $C \in \Gamma(G)$ is a source, we define the {\em weight} of $C$ to be the number of vertices reachable from $C$ (including vertices of $C$) that are not reachable from any other source of $G$. More formally, let $$ \reach'(C) = \{ j \in \reach(C) \,|\, j \not\in \reach(C'), \textrm{for all sources } C' \in \Gamma(G), C' \neq C\}.$$ Then the weight of a source $C \in \Gamma(C)$ is denoted by $w(C) = |\reach'(C)|$ (we do not define weight for components that are not sources). Let us call a source $C \in \Gamma(G_t)$ `light' if its weight $w(C)$ is at most $k=\frac{\alpha n}{2\beta}$ and `heavy' otherwise. By the definition of weight, there could be at most $n/k=2\beta/\alpha$ heavy sources. 

We will argue that, in each step, as we move from $G_t$ to $G_{t+1}$, the number of light sources must decrease by a factor of $(1- \alpha/4)$ with positive probability. For that purpose, suppose there are $m_t$ sources in $G_t$ and among them $m_t'$ are light.  Fix some light source and pick a representative vertex $i$ from it. Since $i$ is contained in a light source, $|\reach'(C(i))|\le  \frac{\alpha n}{2\beta}$. When going to $G_{t+1} = G_t \cup G$, $i$ gets an incoming edge from outside the set $\reach'(C(i))$ with probability at least $\alpha/2$ by Lemma~\ref{lem-alphabetaspread}. If this happens then in $G_{t+1}$, this source will either stop being a source or merge with another source.

Picking a representative for each light source in $G_t$, we see that the expected number of representatives $i$ which get a new incoming edge from outside $\reach'(C(i))$ is at least $(\alpha/2)m_t'$. Hence, this quantity is obtained with positive probability. Now, if at least $(\alpha/2)m_t'$ light sources `merge' with another source or stop being a source in $G_{t+1}$ then the total number of light sources must decrease by at least $(\alpha/4)m_t'$ (the worst case being that $(\alpha/4)m_t'$ disjoint pairs of light sources merge with each other). Hence, with positive probability we get that $m_{t+1}' \leq m_t' \cdot (1 - \alpha/4)$. Therefore, since the samples in each step $t$ are independent, there is also a  positive probability that $m_t' \leq n \cdot (1 - \alpha/4)^t$ and $m_t\le m'_t+2\beta/\alpha$. This completes the proof.
\end{proof}

\subsection{Proof of upper bound from Theorem~\ref{thm-2-spanoids}}

\begin{thm}[Rank of  2-LCSs]\label{thm-2LCSs}
	Let $\cS$ be a $2$-LCS on $[n]$ with error-tolerance $\delta$. Then $\rank(\cS) \leq O(\frac{1}{\delta}\log_2 n)$.
\end{thm}
\begin{proof}
 We will work with the (equivalent) set formulation: let $\cF = \{S_1,\ldots,S_n\}$ be a set system representing the spanoid $\cS$ as in Lemma~\ref{lem-represent}. 
	
	We start by defining an $(\alpha,\beta)$-spread distribution $\mu$ on $\DG(n)$ as follows: To sample a graph $G$ from $\mu$ we first pick $\ell \in [n]$ uniformly at random. Then we add a directed edge from $j$ to $i$ for every $i,j$ such that $\{j,\ell\} \in M_i$.  In this case we have $S_j \cap S_\ell \subseteq S_i$ and so, after restricting to $S_\ell$ we have $S_j\cap S_\ell \subseteq S_i \cap S_\ell$.  
	
\begin{claim}
$\mu$ is a $(2\delta,1)$-spread distribution.
\end{claim}
\begin{proof}
For any fixed $i\in [n]$, $i$ will get an incoming edge if $\ell$, which is randomly chosen from $[n]$, belongs to $M_i$. Since $M_i$ has at least $\delta n$ edges, this will happen with probability at least $2\delta$. Now fix any $i,j \in [n]$, $(j,i)$ will be an edge iff $\ell$ is equal to the the vertex that matches $j$ in the matching $M_i$, this happens with probability at most $1/n$. If $j$ is not matched in $M_i$, the probability is zero.
\end{proof}
  Consider the graph process $G_t^\mu$  and let $S_{\ell_1},\ldots,S_{\ell_t}$ be the sets chosen in the $t$ iterations of sampling from $\mu$. If $i \in \reach(j)$ in the graph $G_t^\mu$, this means that, after restricting to the intersection   $S = S_{\ell_1}\cap \ldots \cap S_{\ell_t}$, the set $S_j$ is contained in $S_i$ (i.e., $S_j \cap S \subseteq S_i \cap S$).  By Lemma~\ref{lem-directed}, after $t=O(\frac{1}{\delta}\log_2 n)$ steps, the graph process $G_t^\mu$ will contain  $r = O(1/\delta)$ sources. Pick a representative $S_{a_1},\ldots,S_{a_r}$  from each of these sources. Then,  the intersection of the $t+r = O(\frac{1}{\delta}\log_2 n)$ sets $S_{\ell_1},\ldots,S_{\ell_t}$ and $S_{a_1},\ldots,S_{a_r}$   is contained in all $n$ sets $S_1,\ldots,S_n$.  That is because, when restricted to the intersection of $S_{\ell_1},\ldots,S_{\ell_t}$, each set $S_i$ contains one of the sets $S_{a_j}, j \in [r]$.
\end{proof}

\subsection{Proof of upper bound from Theorem~\ref{thm-q-spanoids}}

\begin{thm}[Rank of $q$-query LCSs]\label{thm-qLCSs}
Let $\cS$ be a $q$-LCS with error-tolerance $\delta$ and $q \geq 3$. Then $$\rank(\cS) \leq O\left(\delta^{-\frac{1}{q-1}} \cdot n^{\frac{q-2}{q-1}}\log_2 n \right) .$$
\end{thm}
\begin{proof}
	
Like the $2$-query case, we work with the set representation $\cF = \{S_1,\ldots,S_n\}$ of $\cS$ as in Lemma~\ref{lem-represent}. We follow the same strategy as in the proof of the 2-query case. The difference is that, in this case, we will need to pick many sets to restrict to in each step instead of just one. The first observation is that, if $\{j_1,\ldots,j_q\} \in M_i$ then, restricted to the intersection $S = S_{j_1} \cap \ldots \cap S_{j_{q-1}}$ we have $S_{j_q} \subset S_i$. The second observation is that, if we choose a subset $J \subset [n]$ of size roughly $n^{\frac{q-2}{q-1}}$ then, in expectation, $J$ will contain $q-1$ elements in one of the $q$-subsets of $M_i$ for a constant fraction of the $i$'s. Repeating this a logarithmic number of times and using Lemma~\ref{lem-directed}, as in the proof of Theorem~\ref{thm-2LCSs} will then complete the proof.

We start by defining an $(\alpha,\beta)$-spread distribution $\mu$ on $\DG(n)$. To sample a graph $G$ from $\mu$ we first pick a random set $J \subset [n]$ such that each $j \in [n]$ is chosen to be in $J$ independently with probability $(\delta n)^{-1/(q-1)}$. By Markov's inequality we have that
\begin{equation}\label{eq-probjlarge}
	\prob\left[ |J| \geq 4 \cdot \delta^{-\frac{1}{q-1}}n^{\frac{q-2}{q-1}}\right] \leq 1/4.
\end{equation}

For each $i\in [n]$ and  each $q$-subset $T \in M_i$ we select $q-1$ elements of $T$ arbitrarily and refer to them as the {\em distinguished $(q-1)$-subset of $T$}. We now argue that, for each $i \in [n]$, there is relatively high probability that $J$ will contain the distinguished $(q-1)$-subset of at least one $q$-subset in $M_i$. 
\begin{claim}\label{cla-Jhitsi}
Let $E_i$ denote the event that $J$ contains the distinguished $(q-1)$-subset from at least one $q$-subset in $M_i$. Then, for each $i \in [n]$ we have that $\prob[ E_i ] \geq 1/2. $ 
\end{claim}
\begin{proof}
$J$ will contain the distinguished $q-1$ elements in a specific $q$-subset with probability $(\delta n)^{-1}$. Since the $\delta n$ $q$-subsets in $M_i$ are disjoint, the probability that $J$ will not contain any of the distinguished $(q-1)$-subsets is at most $(1 - (1/\delta n))^{\delta n} \leq 1/2$. 
\end{proof}

We are now ready to  define the edges in the graph $G$ sampled by $\mu$. First we check if $|J| \geq 4 \cdot \delta^{-\frac{1}{q-1}}n^{\frac{q-2}{q-1}}$. If this is the case then  $\mu$ outputs the empty graph (by Eq.\ref{eq-probjlarge} this happens with probability at most $1/4$). Otherwise for each $i \in [n]$ we check to see if $J$ contains the distinguished $(q-1)$-subset from one of the $q$-subsets of $M_i$. If there is at least one such $q$-subset, we pick one of them uniformly at random. Suppose the $q$-subset we chose is $\{j_1,\ldots,j_q\}$ and that the distinguished elements are the first $q-1$. Then we add the directed edge  $j_q \rightarrow i$ to the graph $G$. By the above discussion, we know that, restricted to the intersection of all sets indexed by $J$ the set $S_{j_q}$ is contained in $S_i$ (hence the directed edge representing set inclusion).

\begin{claim}
$\mu$ is $(1/4,1/\delta)$-spread.
\end{claim}
\begin{proof}
By Claim~\ref{cla-Jhitsi}, and since the probability that $|J|$ is too large is at most $1/4$ we see that any fixed $i\in [n]$ will get an incoming edge in $G$ with probability at least $\alpha=1/4$. For any fixed $i,j\in [n]$, since the distribution of the special $q$-subset which is contributing an edge to $i$ is uniform in $M_i$ (conditioned on $J$ containing a $q$-subset from $M_i$), we can conclude that $(j,i)\in E(G)$ with probability at most $1/(\delta n)=\beta/n$. This proves the claim.
\end{proof}

Now, applying Lemma~\ref{lem-directed}, we get that, after $t = O(\log_2 n)$ steps, the graph process $G_t^\mu$ will contain at most $O(1/\delta)$ sources with positive probability. Let $J_1,\ldots,J_t$ be the sets chosen in the different steps of the process and, w.l.o.g, remove any of them that were too big (i.e., when the graph sampled by $\mu$ was empty). Hence, all of the sets satisfy $|J_i| \leq 4 \cdot \delta^{-\frac{1}{q-1}}n^{\frac{q-2}{q-1}}$. Now, let $S$ be the intersection of all sets $S_j$ such that $j$ belongs to at least one of the sets $J_i$. Then, restricted to $S$, each of the sets $S_i$ contains one of the sources in the graph $G_t^\mu$. Hence, if we add to our intersection a representative form each of the sources, we will get a set that is contained in all the sets $S_j$. The total number of sets we end up intersecting is bounded by $$ O(1/\delta) + \sum_{i=1}^t |J_i| = O\left(\delta^{-\frac{1}{q-1}} \cdot n^{\frac{q-2}{q-1}}\log_2 n \right).$$ This completes the proof of the theorem.
\end{proof}

\section{Constructing $q$-LCSs with high rank} \label{sec-construction-qLCS}

  In this section we prove the lower bound part of Theorem~\ref{thm-q-spanoids} (the lower bound for the $2$-query case follows from the Hadamard code construction). We will in fact generate this spanoid at random by picking, for each $i \in [n]$, a random $q$-matching $M_i$ on $[n]$ and, for each $q$-subset $T \in M_i$ add the rule $T \spans i$. The resulting spanoid will thus have, by design, the structure of a $q$-LCS. The  reason why this spanoid should have high rank (with high probability) relies on the following observation. Suppose $A \subset [n]$ is a set that spans $[n]$. This means that there is a sequence of derivations $T_i \spans i$ with each $q$-subset $T_i$ in the matching $M_i$ that eventually generates all of $[n]$. We can limit ourselves to the first $C\cdot |A|$ such derivations for some large $C$. These derivations generate a set $A'$ of size $(C+1)|A|$ (including the original $A$ and the $C|A|$ newly derived elements). Now, the set $A'$ must contain all of the $q$-subsets $T_i$ for $C|A|$ values of $i$. However, the union of randomly chosen $C|A|$ $q$-subsets will generally have size much larger than $(C+1)|A|$ (closer to $q\cdot C|A|$).

\begin{thm}[Existence of  high rank $q$-LCSs]\label{thm-randomspanoid}
For any integer $q \geq 3$ and all sufficiently large $n$ the following holds.
Consider the following distribution generating a spanoid $\cS$ on base set $[n]$. For each $i \in [n]$ pick a $q$-matching $M_i$ of size $\lfloor n/2q \rfloor$ uniformly at random and add the rule $T \spans i$ for all $T \in M_i$.  Then, with  probability approaching one,  $\rank(\cS)$ is larger than $r = c n^{\frac{q-1}{q-2}}/\log_2(n)$, where $0<c<1$ is  an absolute constant. 
\end{thm}
\begin{proof}
Let $m = r \cdot \log_2(n) = c n^{\frac{q-1}{q-2}}$. If the rank of $\cS$ is at most $r$ then there exists a set $A \subset [n]$ of size $r$ that spans (using the rules obtained from the $n$ random matchings $M_1,\ldots,M_n$) the entire base set $[n]$. We will upper bound the probability that such a set exists by bounding the smaller event given by the existence of a set of $m$  rules  that can be applied one after another starting with the original set $A$. That is, let $\cE$ denote the event that there exists a set $A$ of size $r$ on which one can sequentially apply $m$ rules of the form $T_{j_i} \spans j_i$ with each $T_{j_i}$ belonging to the matching $M_{j_i}$ and for $m$ different values $j_1,\ldots,j_m \in [n]$ arriving at the final set $\hat A = A \cup \{j_1,\ldots,j_m\}$. If $A$ spans $[n]$ then clearly the event $\cE$ must hold and so, it is enough to show that $\cE$ happens with probability approaching zero.

  We will present the  event $\cE$ as the union of (possibly overlapping) smaller events and then use the union bound, bounding  the probability that each one occurs and multiplying by the number of bad events. Given a set $A \subset [n]$ of size $r$,  a tuple of $m$ indices $\hat J = \{j_1,j_2,\ldots,j_m \}$ and a family of $q$-subsets $\hat T = \{T_{j_1},\ldots,T_{j_m}\}$ with $T_{j_i} \in M_{j_i}$ denote by $\cE(A,\hat J, \hat T)$ the event in which the set $A$ spans the set $\hat A = A \cup \hat J$ using the rules $T_{j_i} \spans j_i$ applied in  order with $i$ going from $1$ to $m$.  For every fixing of $A,\hat J, \hat T$ we can bound $$\prob[ \cE(A,\hat J,\hat T)] \leq  \prod_{i=1}^m \prob[T_{j_i} \subset \hat A].$$ W.l.o.g suppose we sample the random matchings iteratively, picking a new $q$-subset at random among the available elements not covered by any previously chosen $q$-subsets in the current matching.   Since the number of $q$-subsets in each matching is $\lfloor n/2q \rfloor$ we have, at each step, at least $n/2$ available elements to chose from and so $$ \prob[ T_{j_i} \subset \hat A] \leq \frac{{m+r \choose q}}{{n/2 \choose q}} \leq \left(\frac{4m}{n}\right)^q. $$ Taking the product over all $m$ $q$-subsets in $\hat T$ we get
  \begin{equation*}\label{eq-probEAJT}
  \prob[ \cE(A,\hat J, \hat T) ] \leq \left( \frac{4m}{n} \right)^{qm}.
  \end{equation*}
  To complete the proof we bound the number of tuples $(A,\hat J,\hat T)$ as above by $${n \choose r} \cdot {n \choose m} \cdot  \lfloor n/2q \rfloor^m \leq n^r \cdot (en/m)^m \cdot n^m \leq \left( \frac{6n^2}{m} \right)^m, $$ where the last inequality used the fact that $r/m \leq 1/\log_2(n)$. Putting these bounds together we get that $$ \prob[\cE] \leq \left( \frac{4m}{n} \right)^{qm}\left( \frac{6n^2}{m} \right)^m = \left( \frac{6\cdot 4^q \cdot m^{q-1}}{n^{q-2}} \right)^m$$ which is exponentially decreasing in $m$ for the given choice of $m = c\cdot n^{(q-2)/(q-1)}$ and for $c$ a sufficiently small constant. 
\end{proof}

One could ask for a more explicit construction of an LCS with rank equal to (or even close to) that stated above. We are not able to give such a construction but can relate this problem to a longstanding open problem in explicit construction of expander graphs. A bipartite (balanced) expander of degree $q$, is a bipartite graph with $n$ left vertices $L$ and $n$ right vertices $R$ such that the degree of each vertex is $q$ and such that sets $A \subset L$ of size `not too large' have many neighbors in $R$. More specifically, one typically asks that sets with $|A| \leq n/2$ have at least $(1+\eps)|A|$ right neighbors for some constant $\eps>0$. It is quite easy to see that a random graph  of this form will be a good expander with high probability and, by now, there are also many explicit constructions \cite{HooryLW06}. One can also consider {\em unbalanced} bipartite expanders in which $|L| \gg |R|$. Take, for example, the setting in which $|L| = n^2, |R|=n$ and when the degree of every vertex in $L$ is some constant $q$. A simple probabilistic argument shows that sufficiently small sets in $L$, namely sets of size $|A| \leq n^{\alpha_q}$ with $\alpha_q < 1$ a constant depending on $q$ and approaching $1$ as $q$ grows, have many neighbors in $R$ (say, at least $2|A|$). However, no explicit constructions of such graphs are known (for any constant $q$ and any $\alpha_q > 0$). The property we needed in our random construction of  LCSs  can be thought of as an `easier' variant of the expander construction problem. Given $q$-matchings $M_1,\ldots,M_n$ each of size $\delta n$ consider the bipartite graph with $L = [n]\times [\delta n]$ and $R = [n]$. We identify each vertex $(i,j) \in L$ with the $j$'th $q$-subset $T_{ij}$ of $M_i$ and connect it to the $q$ neighbors in $R$ given by that $q$-subset. For our proof to work we need the property that there is no small set containing many $q$-subsets from different matchings. This corresponds to asking for the above graph to be an expander for a restricted family of sets, namely to sets that have at most one vertex $(i,j)$ for a given $i$ (with each subgraph $(i,*)$ defining a matching).

\section{Functional-rank vs. spanoid rank}
\label{sec-frank-vs-rank}

In this section we analyze the five element spanoid $\Pi_5$ described in the introduction (Figure~\ref{fig-pentagon}) and show that its rank is strictly larger than its functional rank. Along the way we formulate the LP relaxation $\LPcover(\cS)$ which lower bounds the functional rank in general and another linear program $\LPentropy(\cS)$ which upper bounds the functional rank.


First we give the lower bound by constructing a consistent code over an alphabet of size $4$ with $32 = 4^{2.5}$ codewords. 
\begin{claim}\label{cla-pentagon-lower}
Let $\Pi_5$ be the pentagon spanoid defined in Figure~\ref{fig-pentagon}. Then, $\frank(\cM_5) \geq 2.5$.	
\end{claim}
\begin{proof}
We will construct a consistent code over the alphabet $\Sigma = \{0,1\}^2$. Each codeword will be indexed by an element of $\{0,1\}^5$. The codeword $w(x) \in \Sigma^5$ corresponding to $(x_1,x_2,x_3,x_4,x_5) \in \{0,1\}^5$ will be $( (x_5,x_2),(x_1,x_3),(x_2,x_4),(x_3,x_5),(x_4,x_1))$ as shown in Figure~\ref{fig-pentagon-code}. In other words, we place the bits $x_1,\ldots,x_5$ on the vertices of the cycle and then assign to each vertex the symbol of $\Sigma = \{0,1\}^2$ comprised of the bits of its two neighbors on the cycle. It is now straight forward to verify that one can compute the coordinate $w_i, i \in [5]$ from the two coordinates spanning it in $\Pi_5$. For example, the span rule $\{1,2\} \spans 4$ requires us to compute $w_4 = (x_3,x_5)$ from $w_1 = (x_5,x_2)$ and $w_2 = (x_1,x_3)$, which can be  easily done (by symmetry, this is the situation in all the other rules). 
\end{proof}

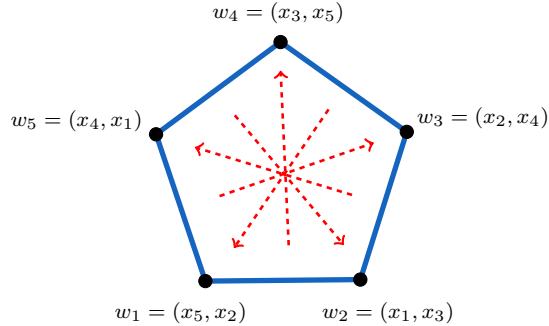
\begin{figure}[!h]
\caption{A consistent code for $\Pi_5$; each coordinate can be recovered from the coordinates of the opposite edge.}
\centering
\definecolor{ffqqqq}{rgb}{1,0,0}
\definecolor{rvwvcq}{rgb}{0.08235294117647059,0.396078431372549,0.7529411764705882}
\begin{tikzpicture}
\label{fig-pentagon-code}
\draw [line width=2pt,color=rvwvcq] (4.16,-2.8)-- (6.22,-2.78);
\draw [line width=2pt,color=rvwvcq] (6.22,-2.78)-- (6.837553878086487,-0.8146432365444852);
\draw [line width=2pt,color=rvwvcq] (6.837553878086487,-0.8146432365444852)-- (5.159223164628246,0.38001404329050925);
\draw [line width=2pt,color=rvwvcq] (5.159223164628246,0.38001404329050925)-- (3.5044038612617054,-0.8470039163194834);
\draw [line width=2pt,color=rvwvcq] (3.5044038612617054,-0.8470039163194834)-- (4.16,-2.8);
\draw [->,line width=1pt,dash pattern=on 2pt off 2pt,color=ffqqqq] (5.26783713739075,-2.3236488686072088) -- (5.16,0);
\draw [->,line width=1pt,dash pattern=on 2pt off 2pt,color=ffqqqq] (4.348459793149962,-1.6691768608425783) -- (6.4,-0.96);
\draw [->,line width=1pt,dash pattern=on 2pt off 2pt,color=ffqqqq] (4.551034462219967,-0.5939728480863998) -- (6,-2.32);
\draw [->,line width=1pt,dash pattern=on 2pt off 2pt,color=ffqqqq] (5.797647810343071,-0.5160595138287057) -- (4.54,-2.36);
\draw [->,line width=1pt,dash pattern=on 2pt off 2pt,color=ffqqqq] (6.109301147373848,-1.6535941939910395) -- (4.02,-1.02);
\begin{scriptsize}
\draw [fill=black] (4.16,-2.8) circle (2.5pt);
\draw[color=black] (3.832877126298422,-3.224936203142211) node {$w_1=(x_5,x_2)$};
\draw [fill=black] (6.22,-2.78) circle (2.5pt);
\draw[color=black] (6.589789150701544,-3.224936203142211) node {$w_2=(x_1,x_3)$};
\draw [fill=black] (6.837553878086487,-0.8146432365444852) circle (2.5pt);
\draw[color=black] (7.853095824763097,-0.6303548461452422) node {$w_3=(x_2,x_4)$};
\draw [fill=black] (5.159223164628246,0.38001404329050925) circle (2.5pt);
\draw[color=black] (5.130167804796906,0.7695105014232468) node {$w_4=(x_3,x_5)$};
\draw [fill=black] (3.5044038612617054,-0.8470039163194834) circle (2.5pt);
\draw[color=black] (2.4484051185337923,-0.66152017984831994) node {$w_5=(x_4,x_1)$};
\end{scriptsize}
\end{tikzpicture}
\end{figure}

%

\subsection{An upper bound on functional rank via $\LPentropy$}
\label{sec-LPentropy}
In this section, we will give a linear programming upper bound for $\frank(\cS)$ using properties of Shannon entropy. We will show that this upper bound matches the lower bound of $2.5$ for $\frank(\Pi_5)$ shown in Claim~\ref{cla-pentagon-lower}, thus proving that $\frank(\Pi_5)=2.5$. We will begin by recollecting some properties of Shannon entropy. 

Given a random variable $X$ supported on some domain $A$, its (Shannon) entropy is defined as $$H(X)=-\sum_{a\in A} \prob[X=a]\log(\prob[X=a]).$$ The Shannon entropy of a random variable measures its information content. The conditional entropy of $X$ given an other random variable $Y$ is defined as: $H(X|Y)=H(X,Y)-H(Y)$. And the conditional mutual information between $X$ and $Y$ given a third random variable $Z$ is defined as: $I(X:Y|Z)=H(X|Z)+H(Y|Z)-H(XY|Z)$. Equivalently, $I(X:Y|Z)=H(X,Z)+H(Y,Z)-H(X,Y,Z)-H(Z)$. Shannon proved that conditional entropy and conditional mutual information are always non-negative~\cite{CoverT91}. These are called basic information inequalities. 

Let $X=(X_1,X_2,\dots,X_n)$ be a random variable made up of $n$ coordinates. Define a function $f:2^{[n]}\to \R_{\ge 0}$ as $f(S)=H(X_S)$ where $X_S=(X_i)_{i\in S}$. Note that $f$ is a \emph{monotone increasing} function because $$f(A\cup B)-f(A)=H(X_{A\cup B})-H(X_A)=H(X_A,X_B)-H(X_B)=H(X_A|X_B)\ge 0.$$ Moreover, $f$ is a \emph{submodular function} i.e. for every $A,B\subset [n]$, $f(A\cup B)+f(A\cap B) \le f(A)+f(B)$. This is because,
\begin{align*}
0&\le I(X_{A\setminus B}:X_{B\setminus A}|X_{A\cap B})\\
&=H(X_{A\setminus B},X_{A\cap B})+H(X_{B\setminus A},X_{A\cap B})-H(X_{A\setminus B},X_{B\setminus A},X_{A\cap B})-H(X_{A\cap B})\\
&=H(X_A)+H(X_B)-H(X_{A\cup B})-H(X_{A\cap B})\\
&=f(A)+f(B)-f(A\cup B)-f(A\cap B).
\end{align*}
In fact, the monotone increasing submodular property of $f$ captures all the inequalities that can be obtained by using the basic information inequalities. But when $n\ge 4$, the entropies $H(X_S)$ satisfy some extra linear inequalities that are not captured by the basic information inequalities. These mysterious inequalities are called non-Shannon type inequalities and a few such inequalities are known~\cite{ZhangY98}, but they are not well understood. The set $$\Gamma_n^*=\{(H(X_S))_{S\subset [n],S\ne \phi}:X=(X_1,X_2,\dots,X_n)\text{ r.v.}\}$$ where $X$ ranges over all $n$ jointly distributed random variables is called the \emph{entropic region} for $n$ random variables. $\Gamma_n^*$ is a convex cone, but neither $\Gamma_n^*$ nor its closure $\overline{\Gamma_n^*}$ are polyhedral for $n\ge 4$~\cite{Matus07} i.e. they are not defined by a finite number of linear inequalities. See~\cite{Yeung08} for more information about non-Shannon type information inequalities and the entropic region.

We are now ready to set up the linear program for upper bounding the functional rank of a spanoid. Let $\cS$ be a spanoid on $[n]$ and let $C\subset \Sigma^n$ be a code over some alphabet $\Sigma$ which is consistent with the spanoid $\cS$ i.e. whenever $A\rightarrow i$ in the spanoid, for every codeword $c\in C$, $c_i$ is determined by $c|_S$. Let $X=(X_1,X_2,\dots,X_n)$ be a random variable with uniform distribution over $C$. Then $A\rightarrow i$ in $\cS$ implies that $H(X_{A\cup\{i\}})=H(X_A)$. The dimension of the code $C$ is $$k=\frac{\log |C|}{\log |\Sigma|}=\frac{H(X)}{\log |\Sigma|}.$$ So upper bounding the functional rank of $\cS$ is equivalent to upper bounding $H(X_1,X_2,\dots,X_n)$ where $X$ is a random variable distributed over $\Sigma^n$ such that $H(X_{A\cup\{i\}})=H(X_A)$ whenever $A\rightarrow i$ in $\cS$. Define $f:2^{[n]}\to \R$ by $$f(S)=\frac{H(X_S)}{\log|\Sigma|}.$$ Clearly $f(\phi)=0$ and $f(\{i\})\le 1$ for all $i\in [n]$. The basic information inequalities are equivalent to saying that $f$ is a \emph{monotone increasing submodular function} i.e. for every subsets $A,B\subset [n]$, $$f(A\cup B)+f(A\cap B)\le f(A)+f(B)$$ and if $A\subset B$, then $f(A)\le f(B)$. Thus the best upper bound we can derive on $H(X)$ using the basic information inequalities is captured by the following linear program. 

\begin{equation}\label{eq-LPentropy}
\begin{aligned}
\LPentropy(\cS)=&\max\ f([n])\\
&f(\phi)=0\\
&f(\{i\})\le 1\ \forall i\in [n]\\
&f(A\cup B)+ f(A\cap B)\le f(A)+f(B)\ \forall A,B\subset [n]\\
&f(A)\le f(B)\ \forall A\subset B\subset [n]\\
&f(A\cup \{i\})=f(A) \text{ whenever } A\spans i \text{ in } \cS
\end{aligned}
\end{equation}

Note that $\LPentropy(\cS)$ is always at most $\rank(\cS)$. This is because, if $A\subset [n]$ is such that $\spn(A)=[n]$, then any feasible $f$ in the LP~(\ref{eq-LPentropy}), should satisfy $f([n])\le |A|$.
The following claim formally states that $\LPentropy(\cS)$ upper bounds the functional rank of the spanoid $\cS$ and lower bounds $\rank(\cS)$, the proof of which follows immediately from the above discussion.
\begin{claim}\label{cla-frank_le_lpentropy}
For any spanoid $\cS$, $\frank(\cS)\le \LPentropy(\cS)\le \rank(\cS).$
\end{claim}

Note that the functional rank could be smaller than $\LPentropy(\cS)$ (though we do not know of an explicit example). This is because the basic information inequalities do not characterize the entropic region of more than 3 random variables. It might be possible to obtain better upper bounds on the functional rank by using non-Shannon type information inequalities. But in the case of the pentagon spanoid $\Pi_5$ defined in Figure~\ref{fig-pentagon}, we will show that $\LPentropy(\cS)$ gives the tight upper bound.

\begin{claim}\label{cla-pentagon-upper}
	$\frank(\Pi_5)= \LPentropy(\Pi_5) = 2.5 < 3 = \rank(\Pi_5)$.
\end{claim}
\begin{proof}
We will show that value of the LP in Equation~(\ref{eq-LPentropy}) is at most $2.5$.
By the spanoid rules, we know that $f(\{i,i+2,i+3\})=f(\{i+2,i+3\})$ for every $i\in [5]$ where the addition is modulo $5$. We want to upper bound $f([5])$. By submodularity of $f$,
$$f(\{1,3,4\})+f(\{1,2,4\} \ge f(\{1,4\})+f(\{1,2,3,4\})$$
$$f(\{1,4\})+f(\{1,5\} \ge f(\{1\})+f(\{1,4,5\}).$$
By the inference rules, $f(\{1,3,4\})=f(\{3,4\})$, $f(\{1,2,4\})=f(\{1,2\})$ and $f(\{1,2,3,4\})=f(\{1,4,5\})=f(\{1,2,3,4,5\})$. Therefore the above two inequalities imply,
$$f(\{1,2\})+f(\{1,5\})+f(\{3,4\})\ge f(\{1\}) + 2f(\{1,2,3,4,5\}).$$
By rotational symmetry, we can obtain five inequalities of this form. Summing them, and observing that every adjacent pair of vertices is counted three times on the l.h.s, we get
$$3\sum_{i} f(\{i,i+1\}) \ge \sum_{i} f(\{i\}) + 10 f(\{1,2,3,4,5\}).$$
Upper bounding $f(\{i,j\})$ by $f(\{i\})+f(\{j\})$, we get:
\begin{align*}
& 6\sum_{i} f(\{i\}) \ge \sum_{i} f(\{i\}) + 10 f(\{1,2,3,4,5\})\\
\Rightarrow & f(\{1,2,3,4,5\}) \le \frac{1}{2} \sum_{i} f(\{i\}) \le \frac{5}{2}
\end{align*}
Therefore by Claim~\ref{cla-frank_le_lpentropy}, $\frank(\Pi_5)\le \LPentropy(\Pi_5)\le 2.5$. By Claim~\ref{cla-pentagon-lower}, $\frank(\Pi_5)\ge 2.5$. This implies the required claim.
\end{proof}

%

\subsection{A lower bound on functional rank via $\LPcover$}
\label{sec-LPcover}
In this section, we will prove lower bounds on functional rank by constructing consistent codes. And the best code one can construct in this way is captured by a very natural linear programming relaxation of the spanoid rank called $\LPcover$. The code constructed in Claim~\ref{cla-pentagon-lower} for the pentagon spanoid $\Pi_5$ can be viewed as an instance of a more general scheme based on a union set-representation of a spanoid. 
\begin{construction}\label{const-lprank}
Let $(S_1,S_2,\dots,S_n)$ be a union set-representation for the spanoid $\cS$ where $S_1,\dots,S_n$ are subsets of a universe $U=\cup_i S_i$ and each $S_i$ is of size at most $\ell$. So whenever $T\spans i$ in $\cS$, $S_i \subset \cup_{t\in T} S_t$. Such a representation can be used to define a consistent code $C$ of dimension $|U|/\ell$ as follows. The codewords are images of the map $C: \{0,1\}^U \to \Sigma^n$ where $\Sigma=\{0,1\}^\ell$ given by $C(x)_i=(x_u)_{u\in S_i}$. It is easy to check that this indeed gives a code consistent with $\cS$ and the dimension of the code is $|U|/\ell$. 
\end{construction}

We will now show that the best consistent code (i.e. of highest dimension) based on this approach can be characterized by an LP.
Note that by Claim~\ref{cla-rank-hittingset}, the rank of a spanoid $\cS$, is the size of the smallest hitting set for $\cO$, the set of all open sets of $\cS$. We can write an LP relaxation for the smallest hitting set for $\cO$. Let $\cO^*\subset \cO$ be the set of minimal open sets, it is enough to hit every set in $\cO^*$.
\begin{equation}\label{eq-LPrank}
\begin{aligned}
\LPcover(\cS)=\min &\sum_{i=1}^n x_i\\
&x_i\ge 0\\
&\sum_{i\in S} x_i \ge 1\ \forall S\in \cO^*.
\end{aligned}
\end{equation}
Since the LP is a relaxation, $$\rank(\cS)\ge \LPcover(\cS).$$ We can round an LP solution to get an integral hitting set by losing a factor of $O(\vc(\cO^*)\cdot \log(\LPcover(\cS)))$ where $\vc(\cO^*)$ is the VC-dimension of $\cO^*$~\cite{EvenRS05,BronnimannG95}. Therefore, $$1\le \frac{\rank(\cS)}{\LPcover(\cS)}\lesssim \vc(\cO^*)\cdot \log(\LPcover(\cS)).$$ Note that $\vc(\cO^*)\le \log |\cO^*|$ always.
By LP duality, we can write a dual LP for $\LPcover(\cS)$ with a dual variable $\lambda_S$ for every $S\in \cO^*$,
\begin{equation} \label{eq-dual-LPrank}
\begin{aligned}
\LPcover(\cS)=\max &\sum_{S\in \cO^*} \lambda_S\\
&\lambda_S\ge 0\\
&\sum_{S\ni i}\lambda_S \le 1\ \forall i \in [n].
\end{aligned}
\end{equation}

The following lemma shows that $\LPcover(\cS)$ is the largest dimension of a consistent code one can obtain using union set-representation of a spanoid as shown in Construction~\ref{const-lprank}. Since the construction of a consistent code for the pentagon spanoid $\Pi_5$ in Claim~\ref{cla-pentagon-lower} is obtained by Construction~\ref{const-lprank} and since we know that $\frank(\Pi_5)=2.5$, it turns out that $\LPcover(\Pi_5)= 2.5.$

\begin{lem}\label{lem-LPrankisSetrank}
Let $C$ be a code consistent with $\cS$ obtained by using Construction~\ref{const-lprank}, then  $\dim(C) \leq \LPcover(\cS)$. Moreover, there exists a code $C$ obtained by using Construction~\ref{const-lprank} giving equality and so $$\frank(\cS)\ge \LPcover(\cS).$$
\end{lem}
\begin{proof}
We will first show that one can get a code whose rank is at least $\LPcover(\cS)$.
Let $\cS$ be a spanoid on $n$ elements with minimal open sets $\cO^*$.
Let $\lambda_S$ be the optimal solution to the dual LP~(\ref{eq-dual-LPrank}) for $\LPcover(\cS)$. Since $\lambda_S$ are rational numbers, let $N$ be least common multiple of all the $\lambda_S$. We will form a multiset $H$ of minimal open sets in  $\cO^*$ where each open set $S\in \cO^*$ appears in $H$ for $\lambda_S N$ number of times. Let $m=|H|=\sum_S{\lambda_S N} = N\LPcover(\cS)$. Let $\Sigma=\{0,1\}^N$.

If we write $H = \{F_1,\ldots,F_m\}$ then a codeword will be indexed by a tuple of bits $x = (x_1,\ldots,x_m) \in \{0,1\}^m$. To construct the corresponding codeword $w(x) \in \Sigma^n$ we need to specify the value of $w(x)_i \in \Sigma = \{0,1\}^N$ for each $i =1,2,\ldots,n$. We set that value to be $w(x)_i = (x_j \,|\, i \in F_j)$ (ordered in increasing order of $j$). That is, we associate a bit with each open set in $H$, and assign the value of a coordinate $i$ to be the list of bits for all sets in $H$ containing $i$. Note that the alphabet size is bounded by $N$ because $|\{j: i\in F_j\}|=\sum_{S\ni i}\lambda_S N\le N$.

We now need to show that, if $T \spans i$ is a rule of $\cS$ then one can recover $w(x)_i$ from $(w(x)_t)_{t \in T}$. Suppose $F_j$ is an open set containing $i$ so that $x_j$ appears in the symbol $w(x)_i$. Then, there must be an element $t \in T$ so that $F_j$ contains $t$ (otherwise, $F_j$ cannot contain $i$). Thus, the value $x_j$ can be computed from $w(x)_t$. Since this holds for any $x_j$ appearing in $w(x)_i$ we are done.

We will now show that the dimension of a code obtained using Construction~\ref{const-lprank} cannot be better than 
$\LPcover(\cS)$. Let $S_1,\dots,S_n$ be the subsets  of size at most $\ell$ from some universe $U$ obtained from Construction~\ref{const-lprank} such that whenever $T\spans i$ in $\cS$, $S_i \subset \cup_{j\in T}S_j$. Such a set system gives a code with dimension $|U|/\ell$. We will show that such a code also gives a feasible solution to the dual LP~(Eqn. \ref{eq-dual-LPrank}). For $u\in U$, let $F_u=\{i: u\in S_i\}$. The complement of $F_u$ is closed since the union of sets which don't have $u$ cannot contain any set which has $u$, therefore $F_u$ is an open set in $\cS$. Let $F_u^*$ be a minimal open set contained in $F_u$. For a minimal open set $F\in \cO^*$, set $$\lambda_F=\frac{1}{\ell}|\{u\in U: F_u^*=F\}|.$$
Let $i\in [n]$, then 
\begin{align*}
\sum_{F\in \cO^*, F\ni i} \lambda_F &= \sum_{F\in \cO^*, F\ni i} \frac{1}{\ell} \sum_{u\in U} \mathbf{1}(F_u^*=F)\\
&= \sum_{u\in U} \sum_{F\in \cO^*,F\ni i} \frac{1}{\ell} \mathbf{1}(F_u^*=F)\\
&= \sum_{u\in U} \frac{1}{\ell} \mathbf{1}(i\in F_u^*)\\
&\le \sum_{u\in U} \frac{1}{\ell} \mathbf{1}(i\in F_u)=\frac{|S_i|}{\ell}=1.
\end{align*}
Thus this assignment is a feasible solution to the dual LP~(Eqn. \ref{eq-dual-LPrank}) with objective value $$\sum_{F\in \cO^*} \lambda_F= \sum_{F\in \cO^*}\sum_{u\in U}\frac{1}{\ell}\mathbf{1}(F_u^*=F)=\sum_{u\in U}\sum_{F\in \cO^*}\frac{1}{\ell}\mathbf{1}(F_u^*=F)=\sum_{u\in U}\frac{1}{\ell}=\frac{|U|}{\ell}.$$
This implies that $\LPcover(\cS)\ge |U|/\ell$.
\end{proof}

The code construction based on sets achieving dimension equal to $\LPcover(\cS)$ constructed in the above above lemma needed very large alphabet. The following lemma shows that by using random sampling, one can get very small alphabet and still achieve dimension close to $\LPcover(\cS)$.

\begin{lem}
Let $\cS$ be a spanoid on $[n]$, then there exists a code $C \subset \Sigma^n$ consistent with $\cS$, obtained as in Construction~\ref{const-lprank} of dimension $r$ over an alphabet $\Sigma$ such that $\log |\Sigma| \lesssim \log n/\log\log n$ and $$r \gtrsim \frac{\log\log n}{\log n}\cdot \LPcover(\cS).$$
\end{lem}
\begin{proof}
Let us assume that $\LPcover(\cS)\ge \log n/\log \log n$, since otherwise the bound is trivial. Form a random subset $H\subset \cO^*$ by sampling each set $S\in \cO^*$ with probability $\lambda_S$ where $\lambda_S$ is the optimal solution to the dual LP~(\ref{eq-dual-LPrank}). Let $Z_S$ be the random variable that $S$ is included in $H$.
$$\E[|H|]=\E[\sum_{S\in \cO^*} Z_S]=\sum_{S\in \cO^*}\lambda_S=\LPcover(\cS).$$
$$\prob\left[|H|\le \frac{1}{2}\E|H|\right]\le \exp\left(-\Omega(\E|H|)\right)=o(1).$$

Let $\Delta(H)$ be the maximum number of sets in $H$ that an element of $\cS$ belongs to. We want to show that $\Delta(H)$ is small with good probability. Fix some $i\in [n]$. The number of subsets in $H$ which contain $i$ is $\sum_{S\ni i} Z_S$.
\begin{align*}
\prob\left[\sum_{S\ni i} Z_S \ge t\right]\le \left(\frac{e\E[\sum_{S\ni i} Z_S]}{t}\right)^t = 
\left(\frac{e\sum_{S\ni i} \lambda_S}{t}\right)^t \le \left(\frac{e}{t}\right)^t
\end{align*}
By union bound, $$\prob[\Delta(H)\ge t]=\prob\left[\exists i\in [n]\ :\ \sum_{S\ni i} Z_S \ge t\right]\le  n\left(\frac{e}{t}\right)^t =o(1)$$
if $t=e\log n/\log\log n$. Therefore there exists an $H\subset \cO^*$ such that $|H|\ge \LPcover(\cS)/2$ and $\Delta(H)\le e\log n/\log \log n$. By imitating the proof of Lemma~\ref{lem-LPrankisSetrank} using this $H$, we get the required result.
\end{proof}

\textbf{Gaps between $\LPcover$ and functional rank:}
Any code that is obtained by using Construction~\ref{const-lprank} from a $q$-LCS with error-tolerance $\delta$ will have dimension at most $O(q/\delta)$. Indeed, suppose $S_1,\dots,S_n\subset U$ are the sets used in the construction. Then each set $S_i$ is contained in $(\delta/q) n$ disjoint $q$-subsets of sets in $\{S_j: j\ne i\}$ such that $S_i \subset \cup_{j \in T}S_j$ for each $q$-subset $T$. Therefore each element of $U$ occurs in $(\delta/q)$ fraction of sets in $S_1,\dots,S_n$. Therefore a typical set in $S_1,\dots,S_n$ contains $(\delta/q)$ fraction of elements from $U$. Therefore the dimension of such a code can be at most $q/\delta$. This also implies that the $\LPcover$ of a $q$-LCS is $O(q/\delta)$. Since the rank of a $q$-LCS (with some constant $\delta$) can be $\tilde{\Omega}(n^{1-\frac{1}{q-1}})$ (Theorem~\ref{thm-q-spanoids}), this shows that $\LPcover$ can be much smaller than the rank of a spanoid. In fact $\LPcover$ can be much smaller than the functional rank ($\frank$) of a spanoid. $\exp(\tilde{O}(\sqrt{\log n}))$-query LCCs (for some constant $\delta$) of length $n$ with linear dimension are known to exist~\cite{KoppartyMRS17}. This implies that functional rank of the corresponding spanoid is $\Omega(n)$ whereas $\LPcover$ is at most $\exp(\tilde{O}(\sqrt{\log n}))$.

\section{Products of Spanoids}
\label{sec-spanoidproduct}
We have shown that functional rank can be strictly smaller than the rank of a spanoid (Claim~\ref{cla-pentagon-upper}), but it was only a constant factor gap. Can we construct spanoids whose $\rank(\cS)$ and $\frank(\cS)$ have a polynomial gap? One way to achieve this is by constructing product operations on spanoids which amplify the gap between $\rank$ and $\frank$. Suppose that given two spanoids $\cS_1$ and $\cS_2$ on $[n_1]$ and $[n_2]$ respectively, we can construct a product spanoid $\cS_1\times\cS_2$ on $[n_1]\times [n_2]$ such that, $\rank(\cS_1\times\cS_2)\ge \rank(\cS_1)\cdot \rank(\cS_2)$ and $\frank(\cS_1\times\cS_2) \le \frank(\cS_1)\cdot \frank(\cS_2)$. Then by starting with the pentagon spanoid and taking the above product several times, one can get a spanoid on $n$ elements with $n^{\Omega(1)}$ gap between $\rank$ and $\frank$. With this motivation, we study a few ways natural ways to define products of spanoids and study how various notions of rank we defined behave under these products.
The following lemma will be useful to compare the rank measures of two spanoids. It shows that a spanoid with more open sets has higher $\rank$, $\frank$, $\LPcover$ and $\LPentropy$. Intuitively, this is because in the spanoid with more open (closed) sets, you can make fewer inferences.
\begin{lem}\label{lem-compare-spanoids}
If $\cS,\cS'$ are two spanoids on $X$ with open sets $\cO_{\cS},\cO_{\cS'}$ respectively. If $\cO_{\cS}\subset \cO_{\cS'}$ then
\begin{align*}
\rank(\cS)&\le \rank(\cS')\\
\LPcover(\cS)&\le \LPcover(\cS')\\
\frank(\cS)&\le \frank(\cS')\\
\LPentropy(\cS)&\le \LPentropy(\cS').
\end{align*}
\end{lem}
\begin{proof}
Since $\rank(\cS)$ is the smallest hitting set for $\cO_{\cS}$ and $\LPcover(\cS)$ is the smallest fractional hitting set for $\cO_\cS$ and $\cO_{\cS}\subset \cO'_{\cS'}$, the first two inequalities easily follow. To show the inequality for functional ranks, we will show that if $A \spans_{\cS'} i$ then $A \spans_{\cS} i$ i.e. $\cS$ has more inference rules than $\cS'$. Thus any code consistent with $\cS$ is also consistent with $\cS'$. This also implies the inequality for $\LPentropy$ because, the corresponding maximization LP for $\cS$ has more constraints than for $\cS'$ and so the maximum is smaller.

The closed sets $\cC_{\cS}$ and $\cC_{\cS'}$ of $\cS$ and $\cS'$ also satisfy $\cC_{\cS}\subset \cC_{\cS'}$. $A \spans_{\cS} i$ iff $i\in \spn_{\cS}(A)$. By Claim~\ref{cla-spanoid-closedsets}, $\spn_\cS(A)=\bigcap_{\{B: B\in \cC_\cS, B \supset A\}} B$ and so $\spn_\cS(A)\supset \spn_{\cS'}(A)$. So if $A \spans_{\cS'} i$ then $A \spans_{\cS} i$.
\end{proof}

\subsection{Product $\cS_1\odot \cS_2$}
\label{sec-product-otimes}
In this subsection, we will define a natural product operation where the open sets in the product are product of open sets and their unions. 
\begin{define}
Given two spanoids $\cS_1,\cS_2$ on sets $X_1,X_2$ with collection of open sets $\cO_1,\cO_2$ respectively, the product spanoid $\cS_1\odot \cS_2$ is an spanoid on $X_1\times X_2$ with open sets given by unions of sets $A\times B$ where $A\in \cO_1, B\in\cO_2$. 
\end{define}
We will now show that under this product, $\LPcover$ is multiplicative.
\begin{lem}
$\LPcover(\cS_1\odot \cS_2)=\LPcover(\cS_1)\cdot \LPcover(\cS_2)$
\end{lem}
\begin{proof}
Let $\cS=\cS_1\odot \cS_2$. The minimal open sets in $\cS$ are cartesian products of minimal open sets in $\cS_1$ and $\cS_2$ i.e. $$\cO^*_{\cS}=\left\{A\times B: A\in \cO^*_{\cS_1}, B\in \cO^*_{\cS_2}\right\}.$$

\begin{equation}
\begin{aligned}
\LPcover(\cS)=\min &\sum_{i,j=1}^n z_{ij}\\
&z_{ij}\ge 0\\
&\sum_{(i,j)\in A\times B} z_{ij} \ge 1\ \forall A\in \cO^*_{\cS_1}, B\in \cO^*_{\cS_2}.
\end{aligned}
\end{equation}
Let $x,y$ be the optimal solutions to the LPs (of the form~(\ref{eq-LPrank})) corresponding to $\LPcover(\cS_1)$ and $\LPcover(\cS_2)$ respectively. Then $z=x\otimes y$ (i.e. $z_{ij}=x_iy_j$) is a feasible solution to the above LP for $\LPcover(\cS)$. Therefore $$\LPcover(\cS)\le \LPcover(\cS_1)\cdot \LPcover(\cS_2).$$
We can also write the dual LP for $\LPcover(\cS)$.
\begin{equation}
\begin{aligned}
\LPcover(\cS)=\max &\sum_{A\in \cO^*_{\cS_1}, B\in \cO^*_{\cS_2}} \lambda_{ST}\\
&\lambda_{AB}\ge 0\\
&\sum_{A\ni i, B\ni j}\lambda_{AB} \le 1\ \forall i,j \in [n].
\end{aligned}
\end{equation}
Let $\alpha$, $\beta$ be the optimal solutions to the dual LPs (of the form~(\ref{eq-dual-LPrank})) corresponding to $\LPcover(\cS_1)$ and $\LPcover(\cS_2)$ respectively. Then $\lambda=\alpha\otimes \beta$ (i.e. $\lambda_{AB}=\alpha_A\beta_B$) is a feasible solution to the above dual LP for $\LPcover(\cS)$. Therefore $$\LPcover(\cS)\ge \LPcover(\cS_1)\cdot \LPcover(\cS_2).$$
\end{proof}

\subsection{Product $\cS_1\otimes\cS_2$}
\label{sec-product-odot}
We will now define a different product operation inspired by the following tensor product operation on codes.

\begin{define}
Given two codes $C_1\subset \Sigma^{n_1}$ and $C_2 \subset \Sigma^{n_2}$, the tensor code $C_1\otimes C_2 \subset \Sigma^{n_1\times n_2}$ is defined as the set of all $n_1 \times n_2$ matrices over $\Sigma$ where each column is a codeword in $C_1$ and each row is a codeword in $C_2$. 
\end{define}

If $C_1$ and $C_2$ are linear codes over some field $\F$ i.e. they are linear subspaces of $\F^{n_1}$ and $\F^{n_2}$ respectively, then the tensor code $C_1\otimes C_2 \subset \F^{n_1\times n_2}$ is exactly the tensor product of the subspaces $C_1\otimes C_2$.


We will now define a product operation on spanoids which mimics the above operation on codes.
\begin{define}
Let $\cS_1$ and $\cS_2$ be spanoids on $X_1,X_2$ respectively. The product $\cS_1\otimes\cS_2$ is a spanoid on $X_1\times X_2$ generated by the following inference rules:
\begin{enumerate}
\item For $A\subset X_1$, $i\in X_1$, if $A\spans i$ in $\cS_1$ then for every $j\in X_2$, $A\times \{j\} \spans (i,j)$.
\item For $B\subset X_2$, $j\in X_2$, if $B\spans j$ in $\cS_2$ then for every $i\in X_1$, $\{i\}\times B\spans (i,j)$.
\end{enumerate} 
\end{define}

The following claim follows easily from the above definitions.
\begin{claim}
Let $C_1$, $C_2$ be codes consistent with $\cS_1$, $\cS_2$ respectively. Then $C_1\otimes C_2$ is consistent with $\cS_1\otimes\cS_2$.
\end{claim}
How does $\cS_1\otimes\cS_2$ compare with $\cS_1 \odot \cS_2$? The following claim shows that $\cS_1\otimes\cS_2$ has more open sets (or closed sets) than $\cS_1 \odot \cS_2$. This shows that all the rank measures are smaller for $\cS_1\odot \cS_2$ than for $\cS_1\otimes\cS_2$.

\begin{claim}
$\cO_{\cS_1\odot \cS_2} \subset \cO_{\cS_1\otimes\cS_2}$. 
\end{claim} 
\begin{proof}
It is enough to show that if $A\in \cO_{\cS_1}$ and $B\in \cO_{\cS_2}$ then $A\times B \in \cO_{\cS_1\otimes\cS_2}$. This is equivalent to showing $(A\times B^c) \cup (A^c \times B)$ being closed in $\cS_1\otimes\cS_2$, which is easy to see given the inference rules.

\end{proof}

\subsection{Product $\cS_1\ltimes \cS_2$}
\label{sec-semidirect-product}
  In this subsection, we define a product operation called the semi-direct product, denoted by $\cS_1\ltimes \cS_2$. Under this product we will show that $\rank$ is multiplicative and $\frank$ is sub-multiplicative. Thus by taking repeated semi-direct product of the pentagon spanoid $\Pi_5$ with itself, we create a spanoid with polynomial gap between its $\rank$ and $\frank$ which proves Theorem~\ref{thm-FrkvsRk}. Additionally, we can also show that $\LPentropy$ is sub-multiplicative under this product which gives a spanoid with polynomial gap between $\LPentropy$ and $\rank$. We will begin with formal definition of semi-direct product.
  \begin{define}[Semi-direct product of spanoids] \label{def-SemiProd}
  Let $\cS_1$, $\cS_2$ be two spanoids on $X_1$, $X_2$ resp.
  Define $\cS_1 \lt \cS_2$ to be the spanoid on $X_1 \times X_2$ generated by the following rules.
  \begin{enumerate}
  \item For $A\subset X_1$, $i\in X_1$, if $A\models i$ in $\cS_1$ then for every $j\in X_2$, $A\times X_2 \models (i,j)$.
  \item For $B\subset X_2$, $j\in X_2$, if $B\models j$ in $\cS_2$ then for every $i\in X_1$, $\{i\}\times B \models (i,j)$.
  \end{enumerate}
  \end{define}
  Note how $\cS_1\ltimes \cS_2$ differs from $\cS_1\otimes\cS_2$ in (1). The semi-direct product is not a symmetric product i.e. $\cS_1\ltimes \cS_2$ may not be isomorphic to $\cS_2 \ltimes \cS_1$. We will first show that $\rank$ is multiplicative under this product.
  \begin{lem} \label{lem-RkLbd}
  $\rank(\cS_1 \lt \cS_2) = \rank(\cS_1) \rank(\cS_2)$.
  \end{lem}
  \begin{proof}
  If $U_1$ is a generating set for $\cS_1$, and $U_2$ is a generating set for $\cS_2$, then $U_1\times U_2$ is a generating set for $\cS_1 \times \cS_2$.
  So $\rank(\cS_1 \lt \cS_2) \le \rank(\cS_1) \rank(\cS_2)$. 
  We only need to prove that $\rank(\cS_1 \lt \cS_2) \ge \rank(\cS_1) \rank(\cS_2)$. The intuition behind this is that, in $\cS_1\ltimes\cS_2$, we can wlog assume that all the $\cS_2$-derivations (i.e.~rules (2) in Definition~\ref{def-SemiProd}) are done before all the $\cS_1$-derivations (i.e.~rules (1) in Definition~\ref{def-SemiProd}). We will write this more formally below.
  
  Let $T$ be a generating set for $\cS_1 \times \cS_2$. We would like to prove that $|T| \ge \rank(\cS_1) \rank(\cS_2)$.
  Let $T_0 = T, T_1, \ldots, T_r = T_1 \times T_2$, and $A_1,\ldots, A_r$ be such that for all $k\in [r]$, we have $T_{k-1} \subsetneq T_k$, $A_k \subset T_{k-1}$, and one of the following is true:
  \begin{enumerate}
  \item $T_k = T_{k-1} \cup (\{i\} \times X_2)$ for some $i\in X_1$, $A_k = A\times X_2$ for some $A\subset X_1$, and $A\models i$ in $\cS_1$.
  \item We are not in Case (1), and $T_k = T_{k-1} \cup \{(i,j)\}$ for some $(i,j)\in X_1 \times X_2$, $A_k = \{i\} \times B$ for some $B\subset X_2$, and $B\models j$ in $\cS_2$.
  \end{enumerate}
If Case (1) is true, we say step $k$ is a $\cS_1$-derivation. If Case (2) is true, we say step $k$ is a $\cS_2$-derivation.
\begin{claim}
 We can choose $T_0,\ldots,T_r$, $A_1,\ldots,A_r$ such that there exists an integer $l$ for which
  \begin{enumerate}
  \item for all $k\le l$, step $k$ is a $\cS_2$-derivation;
  \item for all $k\ge l+1$, step $k$ is a $\cS_1$-derivation.
  \end{enumerate}
\end{claim}  
\begin{proof}
 Suppose there exists some $k$ such that step $k$ is a $\cS_1$-derivation and step $k+1$ is a $\cS_2$-derivation. It is not hard to see that we can swap step $k$ and step $k+1$. Repeatedly applying this until no such $k$ exists, and we get the desired sequences.
  \end{proof}
  Now we return to the proof of $|T| \ge \rank(\cS_1) \rank(\cS_2)$.
  Let $l$ be the integer in the above claim.
  Because steps $k \ge l+1$ are all of type (1), we have
  $$|\{i : \{i\} \times X_2 \subset T_l\}| \ge \rank(\cS_1).$$
  Because steps $k\le l$ are all of type (2), for each $i$ such that $\{i\} \times X_2 \subset T_l$, we have
  $$|T \cap (\{i\} \times X_2)| \ge \rank(\cS_2).$$
  Combining the two inequalities we get $|T| \ge \rank(\cS_1) \rank(\cS_2)$.
  \end{proof}

We will now show that $\frank$ is sub-multiplicative under semi-direct product.
  \begin{lem}\label{lem-FrkUbd}
  $\frank(\cS_1 \lt \cS_2) \le \frank(\cS_1) \frank(\cS_2)$.
  \end{lem}
  \begin{proof}
  Let $\cC \in \Sigma^{X_1\times X_2}$ be a code consistent with $\cS_1 \lt \cS_2$.
  For $i\in X_1$, define code $\cC_2^i \subset \Sigma^{X_2}$ as $\cC_2^i = \{c_{\{i\}\times X_2}: c\in \cC\}.$
  Because $\cC_2^i$ is consistent with $\cS_2$, we have
  ${\log |\cC_2^i|}/{\log |\Sigma|} \le \frank(\cS_2).$
  Let $N=\max_{i\in X_1} |\cC_2^i|$ and for each $i\in X_1$, choose an injection $\phi_i : \cC_2^i \inj [N]$.
  Clearly, we have $$\frac{\log N}{\log |\Sigma|} \le \frank(\cS_2).$$
  Define $\cC_1 \subset [N]^{X_1}$ as $\cC_1 = \{(\phi_i(c_{\{i\}\times X_2}))_{i\in X_1} : c\in \cC\}.$
  Then $|\cC_1| = |\cC|$ and $\cC_1$ is a code consistent with $\cS_1$.
  So we have
  $$\frac{\log |\cC_1|}{\log N} \le \frank(\cS_1).$$
  Combining the inequalities, we get
  $$\frac{\log |\cC|}{\log |\Sigma|} \le \frank(\cS_1) \frank(\cS_2).$$
  \end{proof}


We can now prove Theorem~\ref{thm-FrkvsRk}.
  \begin{proof}[Proof of Theorem~\ref{thm-FrkvsRk}]
  Define $\cS_1 = \Pi_5$, and $\cS_i = \Pi_5 \lt \cS_{i-1}$ for $i\ge 2$.
  Then $\cS_n$ is a spanoid on $5^n$ elements.
  By Lemma \ref{lem-FrkUbd}, $$\frank(\cS_n) \le \frank(\Pi_5)^n = 2.5^n.$$
  By Lemma \ref{lem-RkLbd}, $$\rank(\cS_n) = \rank(\Pi_5)^n = 3^n.$$
  So $$\rank(\cS_n) \ge (5^n)^{\log_5 3 - \log_5 2.5} \frank(\cS_n).$$
  \end{proof}
  
We also show that $\LPentropy$ is sub-multiplicative under semi-direct product.
  \begin{lem} \label{lem-LPeUbd}
  $\LPentropy(\cS_1 \lt \cS_2) \le \LPentropy(\cS_1) \LPentropy(\cS_2)$.
  \end{lem}
  \begin{proof}
Recall the linear program in Equation~(\ref{eq-LPentropy}) that defines $\LPentropy$. 
  Let $f: 2^{X_1 \times X_2} \to \R_{\ge 0}$ be the optimal solution of the linear program which computes $\LPentropy(\cS_1 \lt \cS_2)$, and so $f(X_1\times X_2)=\LPentropy(\cS_1\lt \cS_2)$. Define $f_1 : 2^{X_1} \to \R_{\ge 0}$ as $$f_1(A) = \frac{f(A\times X_2)}{\LPentropy(\cS_2)}.$$
  We claim that $f_1$ satisfies the linear program for $\cS_1$. We will check all the feasibility conditions of the LP given in Equation~(\ref{eq-LPentropy}).
  \begin{enumerate}
  \item $$f_1(\es) = \frac {f(\es)}{\LPentropy(\cS_2)} = 0.$$
  \item For all $i\in X_1$, the function $f_2^i : 2^{X_2} \to \R_{\ge 0}$, defined as $f_2^i(B) = f(\{i\}\times B)$, is a feasible solution for the LP which computes $\LPentropy(\cS_2)$, which is a maximization LP. So
$$f_1(\{i\}) = \frac{f_2^i(X_2)}{\LPentropy(\cS_2)}\le 1.$$
  \item For all $A, B\subset X_1$, we have
  \begin{align*}
  & f_1(A\cup B) + f_1(A\cap B) \\
  &= \frac 1{\LPentropy(\cS_2)} (f((A\cup B) \times X_2) + f((A\cap B) \times X_2)) \\
  &= \frac 1{\LPentropy(\cS_2)} (f((A\times X_2) \cup (B\times X_2)) + f((A\times X_2) \cap (B\times X_2))) \\
  &\le \frac 1{\LPentropy(\cS_2)} (f(A\times X_2) + f(B\times X_2)) \tag{submodularity of $f$}\\
  & = f_1(A) + f_1(B).
  \end{align*}
  \item For $A\subset B\subset X_1$, we have
  \begin{align*}
  f_1(A) = \frac{f(A\times X_2)}{\LPentropy(\cS_2)} \le \frac{f(B\times X_2)}{\LPentropy(\cS_2)} = f_1(B).
  \end{align*}
  \item Let $A\models i$ in $\cS_1$. Then $A\times X_2 \models (i,j)$ for all $j\in X_2$.
  So \begin{align*}
  f_1(A\cup \{i\}) &= \frac 1{\LPentropy(\cS_2)} f((A\cup \{i\}) \times X_2) \\
  &= \frac 1{\LPentropy(\cS_2)}f((A\times X_2) \cup (\{i\} \times X_2)) \\
  & = \frac 1{\LPentropy(\cS_2)}f(A\times X_2) \\
  & = f_1(A).
  \end{align*}
  \end{enumerate}
  So $f_1$ is a feasible solution for the linear program which computes $\LPentropy(\cS_1)$, which is a maximization LP.
Therefore $f_1(X_1) \le \LPentropy(\cS_1).$
So $$\LPentropy(\cS_1\lt \cS_2)=f(X_1\times X_2) = f_1(X_1) \LPentropy(\cS_2) \le \LPentropy(\cS_1) \LPentropy(\cS_2).$$
  \end{proof}

    Thus by repeatedly taking semi-direct product of the pentagon spanoid $\Pi_5$ with itself, we can obtain a polynomial gap between $\LPentropy$ and $\rank$, which is a stronger statement than Theorem~\ref{thm-FrkvsRk} because $\frank(\cS) \le \LPentropy(\cS)$.
  \begin{cor}\label{cor-LPevsRk}
  There exists a spanoid $\cS$ on $n$ elements with $\rank(\cS) \ge n^c \LPentropy(\cS)$
  where $c = \log_5 3 - \log_5 2.5 \ge 0.113$.
  \end{cor}
  

Because $\cS_1\lt \cS_2$ has fewer inferences than $\cS_1\otimes\cS_2$, every closed set in $\cS_1\otimes\cS_2$ is also closed in $\cS_1 \lt \cS_2$ i.e. $\cO_{\cS_1\otimes\cS_2} \subset \cO_{\cS_1 \lt \cS_2}$.
Therefore we have the following relationship between the different products we constructed, 
\begin{equation}
\cO_{\cS_1\odot \cS_2}\subset \cO_{\cS_1\otimes\cS_2} \subset \cO_{\cS_1 \lt \cS_2}.
\end{equation}
By Lemma~\ref{lem-compare-spanoids}, all the rank measures should follow the same order. Therefore, we have the following corollary.

\begin{cor}
For any two spanoids $\cS_1$ and $\cS_2$,
\begin{align*}
\rank(\cS_1\odot \cS_2)\le \rank(\cS_1\otimes\cS_2) &\le \rank(\cS_1\lt \cS_2)=\rank(\cS_1)\cdot\rank(\cS_2)\\
\LPentropy(\cS_1\odot \cS_2)\le \LPentropy(\cS_1\otimes\cS_2) &\le \LPentropy(\cS_1\lt \cS_2)\le\LPentropy(\cS_1)\cdot\LPentropy(\cS_2)\\
\frank(\cS_1\odot \cS_2)\le \frank(\cS_1\otimes\cS_2) &\le \frank(\cS_1\lt \cS_2)\le \frank(\cS_1)\cdot\frank(\cS_2)\\
\LPcover(\cS_1)\cdot\LPcover(\cS_2)=\LPcover(\cS_1\odot \cS_2)&\le \LPcover(\cS_1\otimes\cS_2) \le \LPcover(\cS_1\lt \cS_2).
\end{align*}
\end{cor}
Thus $\LPcover$ is super-multiplicative under semi-direct product, while $\rank$ is multiplicative and the other two are sub-multiplicative.
  

\begin{example}
In contrast to the semi-direct product, $\rank$ is not multiplicative under $\cS_1\odot \cS_2$ and $\cS_1\otimes\cS_2$.
\begin{itemize}
\item This example is found by Yinzhan Xu \cite{Xu18}. Let $\cS$ be the spanoid on a four element set $\{1, 2, 3, 4\}$ with the following closed sets (using the equivalence with intersection closed families):
\begin{align*}
\es, \{1\}, \{4\}, \{2, 4\}, \{3, 4\}, \{1, 2, 3, 4\}.
\end{align*}
Clearly $\rank(\cS) = 2$.
On the other hand, the set $\{(1, 1), (2, 3), (3, 2)\}$ generates $\cS_1\otimes\cS_2$.
So $\rank(\cS_1\otimes\cS_2) \le 3 < 4 = \rank(\cS)^2$.
\item In fact, the spanoid $\Pi_5\otimes \Pi_5$ admits a generating set of size $8$, and therefore $\rank(\Pi_5\otimes \Pi_5) \le 8 < 9 = \rank(\Pi_5)^2$. One such generating set is $$\{(1,1), (1,2), (2,1), (2,2), (3,3), (3,4), (4,3), (5,5)\}.$$
\end{itemize}
\end{example}

\section{Conclusion and open problems}\label{sec-conclusion}
\setcounter{footnote}{0}
Our work introduces the abstract notion of a spanoid in the hope that further study of its properties will lead to progress on LCCs and perhaps in other areas. We list below some concrete directions for future work.
\begin{enumerate}

\item We showed that there exist spanoids, called $q$-LCSs, which ``look like" $q$-LCCs and whose rank matches the best known upper bounds. Can we bypass this `barrier' by using additional properties of LCCs? We have at least two examples where this was possible. One is the result of \cite{KerenidisW04} for LCCs over constant size alphabet and the other is the work in \cite{DvirSW14} for linear $3$-LCCs over the real numbers. The bounds of \cite{KerenidisW04} crucially depend on the alphabet having small size and the bounds in \cite{DvirSW14} exploit properties of real numbers.

\item Understanding the possible gap between functional rank and formal rank of a spanoid is a very interesting question. We proved that there can be a polynomial gap. The next challenge is to find a spanoid on $n$ elements whose $\frank$ is $n^{o(1)}$ and $\rank$ is $n^{\Omega(1)}$. Naturally, $q$-LCSs for constant $q\ge 3$ are plausible candidates for this. If there are no such spanoids, then it would imply the existence of $q$-LCCs of length $n$ and $n^{\Omega_q(1)}$ dimension!\footnote{Possibly over a large alphabet.}

\item Are there general methods (in the spirit of Construction~\ref{const-lprank}, which we show is limited) to achieve high functional rank? 
\item We have constructed spanoid products under which $\rank$ is multiplicative or $\LPcover$ is multiplicative. Can we construct a spanoid product under which $\frank$ is multiplicative?

\item Suppose we start with a functional representation with large alphabet, can we do alphabet reduction without losing too many codewords?

\item We have seen that one way to go past the $\rank$ barrier is to use $\LPentropy$. Can we improve the existing upper bounds on the dimension of $q$-LCCs by upper bounding $\LPentropy$ of $q$-LCSs? Can we use LP duality and construct good feasible solutions to the dual of $\LPentropy$ to prove good upper bounds on $\LPentropy$?
 
\item For a spanoid $\cS$ arising from a matroid, $\LPentropy(\cS)=\rank(\cS)$. This is because the rank function of a matroid is a feasible solution to the LP~(\ref{eq-LPentropy}). Can we separate $\rank$ and $\frank$ for spanoids arising from matroids? One possibility is to use non-Shannon type information inequalities.

\item What are other connections of spanoids to existing theory of set systems, matroids, algebraic equations and other problems described in the introduction?

\end{enumerate}

\subsection*{Acknowledgments}
The second author would like to thank Sumegha Garg for helpful discussions. The third author would like to thank Yury Polyanskiy for helpful discussions.
  



  



  



  
\bibliographystyle{alpha}
\bibliography{refs}

\end{document}